\documentclass[a4]{article}
\usepackage{graphicx} 
\usepackage{authblk}
\usepackage{physics}
\usepackage{braket}
\usepackage{amsfonts}
\usepackage{amsmath}
\usepackage{mathtools}
\usepackage{amsthm}
\usepackage{tikz}
\usepackage{svg}
\usepackage{booktabs}

\usepackage{geometry}
\usepackage{array}
\newcolumntype{P}[1]{>{\centering\arraybackslash}p{#1}}

\newtheorem{theorem}{Theorem}[section]

\newtheorem{lemma}[theorem]{Lemma}

\usepackage[
backend=biber,
style=alphabetic,
sorting=nyt,
maxnames=4,
maxalphanames=4,
giveninits=true,
isbn=false,
url=false
]{biblatex}

\title{Multiplication triples from entangled quantum resources}
\author[1,$\dagger$]{Maxwell Gold}
\affil[1]{Department of Physics, University of Illinois Urbana-Champaign}
\author[2]{Eric Chitambar}
\affil[2]{Department of Electrical and Computer Engineering, University of Illinois Urbana-Champaign}
\affil[$\dagger$]{mjgold2@illinois.edu}
\date{\today}

\begin{document}

\maketitle

\begin{abstract}

An efficient paradigm for multi-party computation (MPC) are protocols structured around access to shared pre-processed computational resources. In this model, certain forms of correlated randomness are distributed to the participants prior to their computation.  The shared randomness is then consumed in a computation phase that involves public communication with efficient round complexity, and the computation is secure in this second phase provided the initial correlations were distributed securely.   Usually the latter requires some strong setup assumptions, such as a trusted dealer and private channels. We present a novel approach for generating these correlations from entangled quantum graph states and yield information-theoretic privacy guarantees that hold against a malicious adversary, with limited assumptions. Our primary contribution is a tripartite resource state and measurement-based protocol for extracting a binary \textit{multiplication triple}, a special form of shared randomness that enables the private multiplication of a bit conjunction. Here, we employ a third party as a Referee and demand only an honest pair among the three parties. The role of this Referee is weaker than that of a Dealer, as the Referee learns nothing about the underlying shared randomness that is disseminated. We prove perfect privacy for our protocol, assuming access to an ideal copy of the resource state, an assumption that is based on the existence of graph state verification protocols. Finally, we demonstrate its application as a primitive for more complex Boolean functionalities such as 1-out-of-2 oblivious transfer (OT) and MPC for an arbitrary $N$-party Boolean function, assuming access to the proper broadcasting channel.

\end{abstract}

\section{\label{sec:INT}Introduction}
The problem of distributed computation within the \textit{pre-processing model} \cite{Bendlin-2011-SE, Damgård-2012-MC}, offers a powerful methodology for constructing information theoretic (or provable) malicious-secure MPC. On the assumption of access to private \textit{shared randomness}, the private evaluation of an arbitrary function, $f$, can be built entirely from information theoretic broadcasting primitives, referred to as \textit{openings}, where parties announce padded information. The trick in computing $f$ with correctness lies in how the parties' randomness is correlated, such that from their openings the parties can locally reconstruct $f$. It is straightforward to prove security in this model, in the sense that a malicious adversary (one who deviates arbitrarily from the protocol) can only learn anything about an honest party's input from their local reconstruction of $f$ \cite{Beaver-1992-EM}. This all hinges upon the assumption of access to such forms of correlated randomness, however, where it is known to be impossible to disseminate these resources in a purely classical setting without resorting to strong information theoretic assumptions, such as a trusted Dealer that does not know $f$ and private peer-to-peer channels, or unproven hardness assumptions \cite{Impagliazzo-1989-Of}.  

In this work we tackle the problem of disseminating multiplication triples from entangled quantum resources. A multiplication triple is unique form of shared randomness that enables private multiplication of party inputs in this model. With quantum resources, we achieve information theoretic malicious-secure privacy based on assumptions much weaker than a trusted Dealer. In this paper, we formally define the problem of \textit{quantum triple distribution} (QTD), which similar to its namesake, works to instantiate a functionality for triple distribution (TD), and present an instantiating protocol that utilizes entangled resources to extract such a correlation. 

We work in a three-party computation (3PC) setting, utilizing a non-collaborating Referee ($\mathbf{R}$) to assist a pair, Alice ($\mathbf{A}$) and Bob ($\mathbf{B}$), in a protocol for disseminating binary multiplication triples. To this end, we introduce a particular tripartite quantum resource, represented as a graph state, $\ket{G_{\wedge}}$, from which a triple can be extracted through local measurements made by the parties. We show that so long as any pair among $\{\mathbf{A},\mathbf{B},\mathbf{R}\}$ act honestly, a private binary multiplication triple can be extracted without communication, from the resource state we introduce. In order to efficiently prove the privacy offered by our resource state we assume an ideal quantum resource is distributed to each party when executing our protocol. The appeal is that there are protocols for distributing graph states with an information theoretic certificate of authenticity \cite{McKague-2014-SG, Baccari-2020-SB, Markham-2020-SP}. In particular, there exist resource efficient graph state \textit{verification} protocols, which have been shown to handle malicious networks of players \cite{Unnikrishnan-2022-Vg}, and support composable security definitions on their own \cite{Colisson-2024-Ag}. This lets us assume access to an ideal functionality $\mathcal{F}_{\ket{G_{\wedge}}}$ that distributes the resource state, when proving the security of our protocol.

A feature in utilizing quantum resources for cryptography is in there ability to encode classical shared randomness. Their known advantage arises from how the no-cloning theorem enables some level of cheat-sensitivity when transmitting information through quantum channels. The canonical example of quantum key distribution (QKD) allows parties to tackle the classical problem of key distribution (KD), and establish shared bipartite secret key in a way that is statistically resilient against an active eavesdropper \cite{Bennett-1984-Qc, Ekert-1991-Qc}, and composable in any cryptographic system \cite{Ben-Or-2005-UC}. By utilizing such resources for provable secure TD, a protocol for QTD removes several assumptions typical of classical analogs. Namely, there is no need for a trusted Dealer who does not know $f$, and no assumptions on private channels are necessary.

Hence, instantiating a protocol for QTD, along with the proper ideal broadcast channel, allows a group to tackle the problem of secure function evaluation from start to finish with information-theoretic malicious-secure privacy guarantees, where assumptions come in from three distinct sources: 
\begin{itemize}
    \item The various assumptions tied to the certification of the resource states.
    \item The assumption of at least an honest pair acting on each tripartite resource.
    \item The assumption of an ideal broadcasting channel for the online phase.
\end{itemize}
The advantage that quantum resource provide in this problem comes from the combined assumptions from the first two bullets, which make up the necessary requirements to build an offline phase from QTD. In this paper, we provide a solution for the second bullet. Furthermore, despite requiring an honest majority when strictly working in a 3PC setting, by further forcing the Referee be honest in our protocol, an $(N+1)$-party setting (not counting the Referee against $N$) can achieve privacy against a malicious majority, given $N>3$, from only QTD and an ideal broadcast channel.

\subsection{\label{sec:INT-related}Related work}
In \cite{Gold-2025-Hp}, we presented a version of our primitive embedded in a larger 3PC protocol that allowed for the evaluation of a limited class of Boolean functions that only admit quadratic bit conjunctions. While this was sufficient for privately computing any function taking inputs from only two parties, we show here how this can be extended to any multiparty function. Furthermore, the primitive itself has been retooled, such that it directly yields a binary multiplication triple, from non-interactive input-independent quantum measurements made on an ideal resource state. 

To the best of our knowledge there is little existing work regarding protocols that focus on efficiently disseminating multiplication triples from entangled quantum states. In \cite{Roehsner-2021-Po}, it was shown how entanglement can be used to construct shared tables for probabilistic one-time programs, which can function as a computational resource in a manner similar to noisy multiplication triples. This is distinct from our work, as correctness is always guaranteed from honesty in our primitive. In a different but related vein, \cite{Yang-2014-Qo} presents a protocol for computing OT from quantum resources, employing an untrusted third party. This further motivates the sense in which multipartite entanglement can enable cryptographic primitives, beyond what can't be done in the bipartite case. 

Below we review some of the wider context for our work.

\medskip\noindent\textbf{Removing assumptions from information theoretic MPC.} 
At a glance, MPC describes the problem of two or more parties who wish to securely evaluate a function, $f$, of their shared inputs, in such a way that no information is revealed about a party's input to an adversary, beyond what is implied by $f$ itself \cite{Yao-1982-Ps}. In what sense it is meant that ``no information is revealed'' to the adversary is subject to the set of cryptographic assumptions one builds into the protocol they design to solve this problem. The strongest guarantees then come when one can \textit{prove} security against the adversary, without any assumptions on computational hardness. In terms of privacy, these information theoretic guarantees come in two forms: \textit{perfect} privacy, where no information is leaked by the protocol whatsoever, and \textit{statistical} privacy, where the amount of information leaked to the adversary is bound by their ability to repeat the protocol many times.

Information theoretic MPC in the context of ideal peer-to-peer channels is well understood, where solutions universal to any $f$ in the case on honest majority are known \cite{Ben-Or-1988-CT, Chaum-1988-Mu}. The amount of communication in this model, however, is particular costly when performing multiplication between distinct parties' inputs. MPC within the pre-processing model circumvents this issue by breaking up the computation into \textit{offline} and \textit{online} phases, where most of the communication occurs in the offline phase for the process of pre-computing triples and other forms of shared randomness. The key efficiency lies in reducing the communication cost to compute multiplication online. The evaluation of a logical conjunction of bits in the online phase requires broadcasting only $\mathcal{O}(2)$ bits across public channels, versus $\mathcal{O}(n)$ bits across private channels in \cite{Ben-Or-1988-CT}.

Before moving on, we note that within the larger scope of MPC's more recent history, a protocol based on secret sharing that only guarantees the \textit{privacy} of the parties' inputs is somewhat limited. The first notions of \textit{verifiable} secret sharing were introduced in \cite{Chor-1985-Vs}. The work of \cite{Bendlin-2011-SE, Damgård-2012-MC} demonstrates how to construct representations of triples that include authentication codes for validating the sharing explicitly, from somewhat homomorphic encryption. In our case, the sharing distributed by our protocol contains just the information necessary to function as a triple, with no means of authenticating the share. Hence, any secret sharing-based MPC that calls on the protocol we present as a primitive will only guarantee \textit{security with abort} (where the adversary can guarantee only they receive an output), and not anything stronger, such as \textit{fairness}, or \textit{guaranteed output delivery} \cite{DovGordon-2015-CM}. 

\medskip\noindent\textbf{Impossibility results for deriving provable cryptography from quantum resources.} It is well known that quantum states are not commitments \cite{Lo-1997-QB, Mayers-1997-US}. In a quantum bit commitment protocol, as an example, if the Receiver is unable to verify that their state is disentangled from the environment, i.e. that it is pure, these commitment schemes allow a cheating Sender to ``fake'' their commitment by maintaining their own pure description of the quantum system. In this way the execution of the protocol appears to the Receiver indistinguishable from an honest execution, as in either case their local description of the quantum system remains unchanged, but the commitment is not binding for the Sender. Other important primitives, such as OT and fair coin tossing, are also known to be insecure against an adversary who \textit{purifies} their actions with unbounded resources \cite{Lo-1997-Iq, Mochon-2007-Qw}. It has been proven that all non-trivial two-party primitives necessarily leak some information in this \textit{quantum honest-but-curious} standard for security \cite{Salvail-2009-PT}. Solutions exist that rely on bounding either the classical or quantum computing power of the adversary \cite{Damgård-2005-CB}, however we make no such restrictions. It follows from the impossibility of OT in the two-party quantum setting that a multiplication triple cannot be securely extracted from only a bipartite quantum system.

In addition to these no-go theorems, multiparty secret sharing protocols employing multipartite entangled quantum states are also know to be susceptible to a so-called \textit{participant attack}. In this setting, information obtained from an honest party revealing the choice of measurement basis used in the protocol allows an adversary who delayed their measurement to obtain the secret fully \cite{Karlsson-1999-Qe, Qin-2007-CH, Walk-2021-SC}. Solutions to such a problem involving carefully addressing when these measurement bases are revealed, or by choosing a quantum state with a particular secret sharing \textit{access structure} that limits what an adversary can learn \cite{Markham-2008-Gs}. Our protocol applies elements of both approaches.

\medskip\noindent\textbf{Graph state verification protocols.} A useful tool in designing cryptographic protocols around multipartite entangled resources, such as graph states, is the existence of resource efficient \textit{verification} protocols, which work to certify the state against the possibility of a corrupt Source that distributes these resources unfaithfully. These protocols function through requesting many copies of the resource graph state, from which they test the set of possible classical correlations that uniquely defined by the structure of the corresponding graph \cite{Takeuchi-2019-Rv, Markham-2020-SP}. With each copy of the resource, recipients randomly choose together between consuming the resource through the designated protocol, and performing rounds of testing. The tests consist of measuring local components of an element of the stabilizer of the graph state, from which the parties always receive a combined measurement outcome of $+1$, when the state is distributed honestly. A sampling argument then allows a verifier within the group of recipients, with who the parties need only communicate publicly, to exponentially bound the probability that a single target round of execution will operate on the correct resource state. A practical assumption built into these protocols is that measurement devices behave honestly, but no restrictions are placed on the Source who distributes the resources unfaithfully. 

Of significant interest, are versions of these verification protocols that additionally tolerate an adversarial presence within the network of recipients, colluding with the corrupt Source. In particular \cite{Unnikrishnan-2022-Vg} demonstrated one such protocol satisfying these requirements within the common reference string (CRS) model. The CRS is used to independently decide in which rounds the resource is tested (where a verifier is also picked at random by the CRS) or consumed by the protocol. More recently, \cite{Colisson-2024-Ag} demonstrated that all graph state verification protocols support composable security definitions, with minor modifications, including the protocol from \cite{Unnikrishnan-2022-Vg}. Together, this implies that we can in principle separate the problem of distributing graph states from any a protocol that consumes that state as a resource, and reduce the necessary assumptions to privately distribute these resources (the first bullet in the above) to that of a CRS. Hence, in this paper we assume an ideal copy of the resource state is already in position of the parties running the protocol, such that the proofs we present are tightly focused on the nature of correlations encoded in independent copies of the resource.

\subsection{\label{sec:INT-contribution}Primary contributions}
We give a protocol for generating a sharing of a binary multiplication triple, utilizing distributed correlations obtained from local quantum measurements made on a graph state. Classical information can be encoded in the state, in the form of relative phases between neighboring qubits in the graph. This relative \textit{phase information} then manifests itself in the measurements outcomes obtained in our protocol, when parties measure local components of a stabilizer of the graph state. Beaver's original notion for securely executing a multiplication gate involved first executing the gate on random inputs and thereafter correcting the outcome to reflect the desired conjunction by public broadcast of the inputs with proper padding \cite{Beaver-1992-EM}. In the same vein, our protocol distributes private measurement outcomes that encodes a random bit conjunction, which is private from all the parties. 

\medskip\noindent\textbf{Delegated computation.} We work in a 3PC setting, utilizing a non-collaborating Referee ($\mathbf{R}$) to assist a pair, Alice ($\mathbf{A}$) and Bob ($\mathbf{B}$), in the protocol. We prove that our protocol provides composable perfect privacy of the distributed sharing, against a malicious adversary corrupting a single party $\mathbf{X}\in\{\mathbf{A},\mathbf{B},\mathbf{R}\}$, assuming an ideal copy of the graph state is distributed. In this sense, our security requirement is that of an honest majority, in this 3PC setting, which we refer to as an \textit{honest pair}. The protocol breaks in trivial ways if the adversary corrupts any two parties in the protocol. Note, however, that we prove that our measurement-based protocol can be implemented composably, such that when triples from our protocol are input into a higher level secret sharing-based $(N+1)$-party MPC, security of the broadcast information is still guaranteed from only an honest pair among the three who obtained the triple, regardless of the number of corrupted parties outside of the primitive. In other words, for a $(N+1)$-party setting employing QTD as a primitive for generating triples, an honest $\mathbf{R}$ ensures the privacy of any honest individual's input. In contrast, if $\mathbf{R}$ is corrupted, then all other parties must remain honest to guarantee the privacy of their inputs.

\medskip\noindent\textbf{Conditional stabilizer tests.} The correlations extracted from the measurement sequence in our protocol must remain private. In typical secret sharing protocols with graph states, the local components of a fixed generating operator of the state's stabilizer are measured, such that the sum of the parties' measurement outcomes directly encodes the secret. In our case, the final choice of measurement basis each party makes at the end of the sequence is conditioned on private information they obtained in prior steps. Regardless of the choice of these local bases, these local measurements together still make-up an element of the stabilizer. In addition to giving us a correlation of a desirable form, we avoid the possibility of a participant attack by keeping these measurement bases private. 

\subsection{\label{sec:INT-outline}Organization of contents}
In Sec. \ref{sec:PRE}, we establish notation and present background on binary multiplication triples, as well as graph states. In Sec. \ref{sec:QTD}, we formally define the problem of QTD from an ideal functionality for TD, and compare solutions to the classical case of achieving TD through a trusted Dealer and private channels. In Sec. \ref{sec:OVR}, we describe in detail a 3PC protocol in which a tripartite graph state encodes a binary multiplication triple in the outcomes of local quantum measurements, and provide a high-level picture of the privacy of this quantum resource. In Sec. \ref{sec:SEC}, we formally prove the security of this protocol against a malicious adversary, and discuss its composability. Finally, in Sec. \ref{sec:APP}, we share applications of QTD as a primitive, demonstrating how functionalities for evaluating referee-assisted OT and $N$-party Boolean functions can be constructed from only our primitive and an ideal public broadcast channel.

\section{\label{sec:PRE}Preliminaries}
We adopt the following consistent set of notation regarding party labels and secret sharings of classical information. Thereafter, background on multiplication triples is reviewed. Lastly, we present a few key concepts surrounding graph states and secret sharing that will be employed in our proofs of correctness and privacy.

\medskip\noindent\textbf{Notation.} Ideal functionalities and protocols and will be labeled by $\mathcal{F}$ and $\Pi$, respectively. As noted above, we denote the set of parties in our 3PC protocol as $\{\mathbf{A},\mathbf{B},\mathbf{R}\}$ and use $\mathbf{X}$ at times to refer to an arbitrary member. When working in a larger $N$-party setting, we denote the set of parties $\{\mathbf{P}_1,\cdots,\mathbf{P}_N,\mathbf{R}\}$, where $\{\mathbf{A},\mathbf{B}\}$ are played be any pair in $\{\mathbf{P}_1,\cdots,\mathbf{P}_N\}$ during each call to our protocol. When discussing the completely classical analog, $\mathbf{D}$, will refer to the trusted Dealer. When dealing with a malicious adversary, we will use the notation $\mathbf{X}^{*}$ to specify that party $\mathbf{X}$ has been corrupted and deviates arbitrary from the protocol. For classical variables of $\mathbb{Z}_2$, such as inputs, messages, private measurements outcomes, and encoded phase information, we will use $p\oplus q$ and $pq$ to denote addition and multiplication modulo 2, respectively, of these values. $[p]$ denotes an additive sharing of such a variable $p$, and $[p]_{\mathbf{X}}$ denotes the share belonging to $\mathbf{X}$, such that $p=\sum_{\mathbf{X}}[p]_{\mathbf{X}}$. 

\medskip\noindent\textbf{Binary multiplication triples.} A binary multiplication triple is defined as triple of random but correlated bits $\{p,q,pq\}\in\mathbb{Z}_{2}^{3}$, for which there exists an additive sharing $[p], [q], [pq]$. Given such a distributed correlation, parties can obtain a sharing of a bit conjunction $[ab]$, with definite inputs $a,b$, by first, simultaneously opening corrections $c_{a} = a \oplus p$ and $c_{b} = b \oplus q$, and then, locally computing
\begin{equation}
    [ab]_{\mathbf{X}} \leftarrow c_{a}c_{b} \oplus c_{a}[q]_{\mathbf{X}} \oplus c_{b}[p]_{\mathbf{X}} \oplus [pq]_{\mathbf{X}}. \label{eq:Trip-Id}
\end{equation}
Assuming the triple was securely distributed by a trusted Dealer, $\mathbf{D}$, who does not know $f$, and no parties hold any prior knowledge about either input, it is straightforward to prove the information-theoretic security of this protocol, as
\begin{equation}
     P(a=0|c_{a})=P(b=0|c_{b}) =1/2. \label{eq:Trip-Sec}
\end{equation}
Given an ideal broadcast channel, where an eavesdropper cannot modify messages sent, these openings serve as information theoretic primitives that can trivially be run concurrently or sequentially, depending on the round complexity required to compute $f$.

In our work, we consider a special case of a triple, where we restrict that $\mathbf{D}$ set $[p]_{\mathbf{X}}\leftarrow 0$ for all parties, except the party $\mathbf{A}$ possessing the input $a$, where they instead assign $[p]_{\mathbf{A}}\leftarrow p$. Security still holds as
\begin{equation}
    P(b=0|c_{b},p) = P(q=0|c_{b},p)=1/2. \label{eq:Trip-Sec-A}
\end{equation}
The same argument applies when $\mathbf{D}$ also gives the party $\mathbf{B}$ possessing $b$ private access to $q$. While there is no advantage to distributing triples this way, other than the minor convenience of reducing the number of individual broadcast messages required for each opening, we show that this is the natural way a triple is extracted from the resource in our protocol. 

\begin{figure}[]
    \centering
    \framebox{
        \begin{minipage}[t]{0.99\textwidth} 
            \centerline{Functionality $\mathcal{F}_{\text{TD}}$}
            \medskip
            Sample $p,q\leftarrow \mathbb{Z}_2$. Distribute $p$ to $\mathbf{A}$, $q$ to $\mathbf{B}$, and shares $[pq]_\mathbf{X}$ to all $\mathbf{X}\in\{\mathbf{A},\mathbf{B},\mathbf{R}\}$.
        \end{minipage}
    }
    \caption{An ideal functionality $\mathcal{F}_{\text{TD}}$ for distributing binary multiplication triples in a three-party setting.}
    \label{fig:F-TD}
\end{figure}

\medskip\noindent\textbf{Graph states.} Let $G=(V,E)$ be an arbitrary graph, where $V=\{1,\cdots,n\}$ and $E\subseteq V\times V$. The $n$-qubit graph state $\ket{G}$ is a unique eigenstate of a set of \textit{correlation operators}, $K_i$, that generate the stabilizer, $S$, of the state,
\begin{equation}
    S=\left\langle \left\{ K_{i} = X_i \prod_{j \in V} Z_j^{\Gamma(i,j)} \,\middle|\, \forall i\in V  \right\}\right\rangle, \label{eq:S-G-arb}
\end{equation}
where $\Gamma(i,j)$ is the typical adjacency matrix of $G$ \cite{Hein-2004-Me}. Explicitly, we fix $\ket{G}$ to be the state, such that $\Tr[K_i \ketbra{G}]=1$ for every $i\in V$, and
\begin{equation}
    \ketbra{G} = \frac{1}{2^{n}}\prod_{i\in V}(\mathbb{I}+K_i). \label{eq:G-rho}
\end{equation}
Note that every state with a unique stabilizer is equivalent to a graph state by a set of local unitary operations \cite{VanDenNest-2004-Gd}.

We can encode classical information in a graph state by defining a state $Z^{\vec{r}}\ket{G}$, where $\vec{r}\in \mathbb{Z}_{2}^{n}$ and $Z^{\vec{r}}=\prod_{i\in V}Z^{r_{i}}$. We refer to $r_{i}$ as the \textit{phase information} encoded in qubit $i$ of $\ket{G}$, and each correlation operator transforms accordingly as $K_{i}\to (-1)^{r_{i}}K_{i}$. 

\begin{lemma}[Stabilizer tests.]\label{lem:G-corr}
Let $\Lambda \subseteq V$, and $K_{\Lambda}=\prod_{i\in\Lambda} K_{i}$ be an operator, such that $(-1)^{\bigoplus_{i\in\Lambda}r_{i}}K_{\Lambda}\in S$. Suppose that $K_{\Lambda}$ acts non-trivially on only qubits in $\Omega$, such that $\Lambda\subseteq\Omega\subseteq V$. By measuring each qubit in $\Omega$ of $Z^{\vec{r}}\ket{G}$ individually in the basis specified by $K_{\Lambda}$, one obtains a set of correlated measurement outcomes $\{m_{i}\in\mathbb{Z}_2\}_{i\in\Omega}$, such that
\begin{equation}
    \bigoplus_{i\in \Omega}m_i =  \bigoplus_{i\in \Lambda} r_i, \label{eq:S-Meas}
\end{equation}
with unity probability. 
\end{lemma}
This collection of local measurements makes up a \textit{stabilizer test}, as it tests whether $K_{\Lambda}\in S$ up to an overall phase, $\bigoplus_{i\in\Lambda}r_i$. Eq. \ref{eq:S-Meas} is the basis of classical secret sharing protocols using graph states. By carefully choosing which parties hold the qubits in $\Omega$, an additive sharing of the phase information encoded in the qubits in $\Lambda$ can be distributed. Furthermore, given a graph state where $\vec{r}=0$, these measurements are also the intuition behind verification protocols.

In addition to these correlated measurement outcomes, for larger graph states there are often qubits in the neighborhood of those in $\Omega$, which $K_{\Lambda}$ acts trivially on, and hence, aren't measured in obtaining the correlations associated with $K_{\Lambda}$. In general, local measurements transform the graph state about the neighborhood of the measured qubits. This can transfer or delete the phase information of the measured qubit, depending on the basis of the measurement.

(For a more detailed discussion on graph states, see Appendix \ref{sec:APP-Graph states}.)

\section{\label{sec:QTD}Quantum triple distribution}
The goal of this work is to distribute the appropriate shares of binary multiplication triple $\{p,q,pq\}$ in a 3PC setting, in such a way that no one single party knows the value of the bit $pq$. In this sense we say that the triple disseminated from this process is \textit{private}, as the correlation encoded in the sharing remains a uniformly random unknown for all parties. Fig. \ref{fig:F-TD} describes the ideal functionality for such a process, between a triple of parties $\{\mathbf{A},\mathbf{B},\mathbf{R}\}$.

\medskip\noindent\textbf{Triple distribution from classical resources.}
In a fully classical setting, where parties posses no quantum resources, it is impossible instantiate $\mathcal{F}_{\text{TD}}$ in a provably secure way, without strong assumptions involving honesty and private channels. One such case is to assume a trusted Dealer, $\mathbf{D}$, that does now know $f$ of the encompassing MPC. A protocol for classical TD, $\Pi_{\text{CTD}}$, then involves $\mathbf{D}$ locally generating $\{p,q,pq\}$ and sending $\{p,[pq]_{\mathbf{A}}\}$ to $\mathbf{A}$ and $\{q,[pq]_{\mathbf{B}}\}$ to $\mathbf{B}$. To more closely match the functionality introduced above, there is no additional cost to further allowing $\mathbf{D}$ to keep a share $[pq]_{\mathbf{D}}$ as well. In Fig, \ref{fig:CTD-prob}, we depict a diagram for the network of parties in this setting. The notion of a private triple is replaced here by $\mathbf{D}$'s lack of information for $f$. Hence, none of the parties who would consume the triple in an MPC uniquely know $pq$, at cost of incorporating a single party that does not know $f$.

\begin{figure}[]
    \centering
    \includegraphics[width=0.5\textwidth]{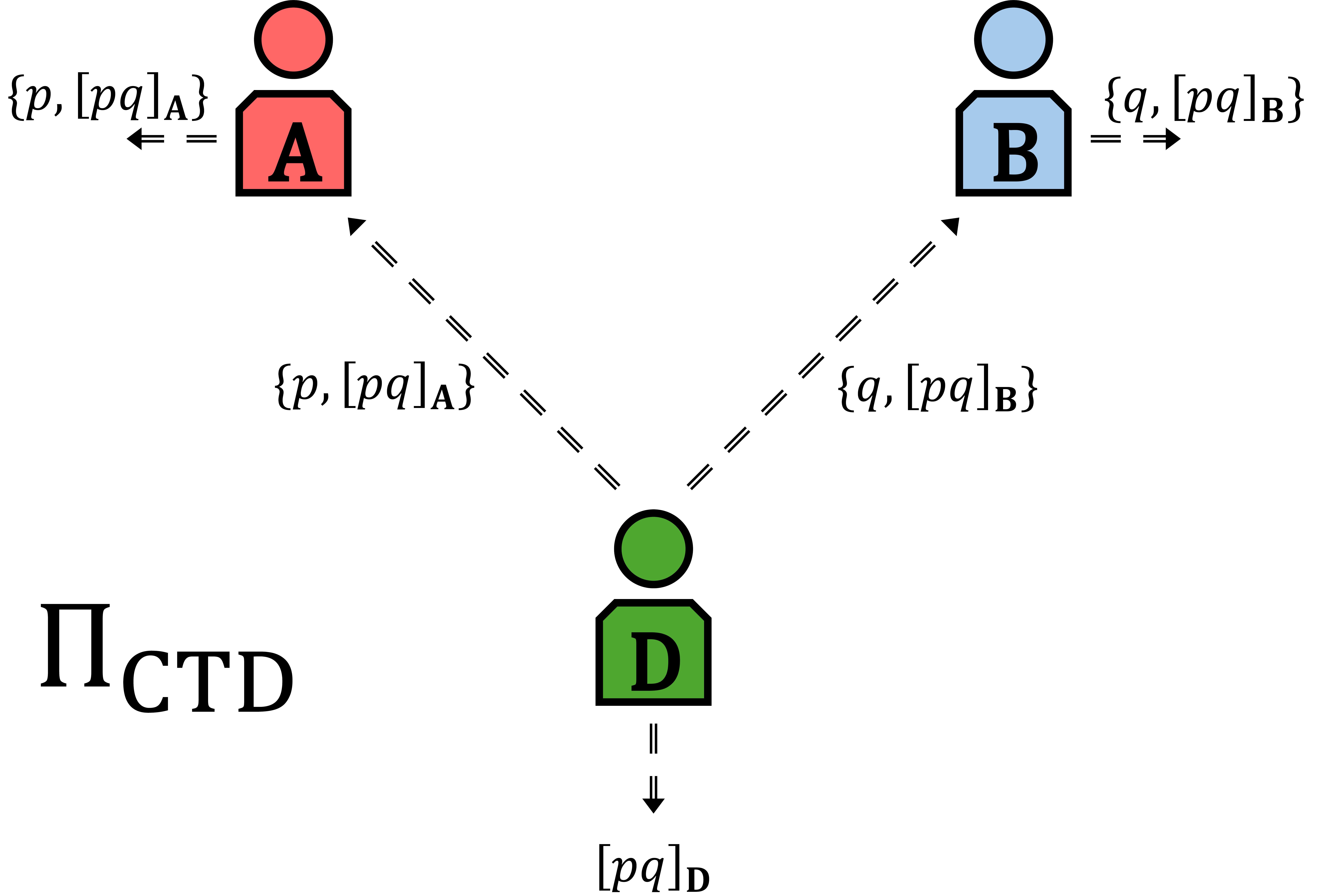}
    \caption{A network diagram representing a classical protocol for implementing $\mathcal{F}_{\text{TD}}$, in a modified setting. Dashed double-lined arrows indicate the distribution of private classical information. In order to instantiate TD with information theoretic privacy from only classical resources, an trusted $\mathbf{D}$ is required, who generates $\{p,q,pq\}$ locally and sends the appropriate bits over private channels to $\mathbf{A}$ and $\mathbf{B}$. As $\mathbf{D}$ must remain honest here, they replace $\mathbf{R}$ in $\mathcal{F}_{\text{TD}}$, and can optionally choose to retain a share of $[pq]$. In this setting it is impossible to generate the triple in such a way that no one party knows the value of $pq$.}
    \label{fig:CTD-prob}
\end{figure}

\begin{figure}[]
    \centering
    \includegraphics[width=0.99\textwidth]{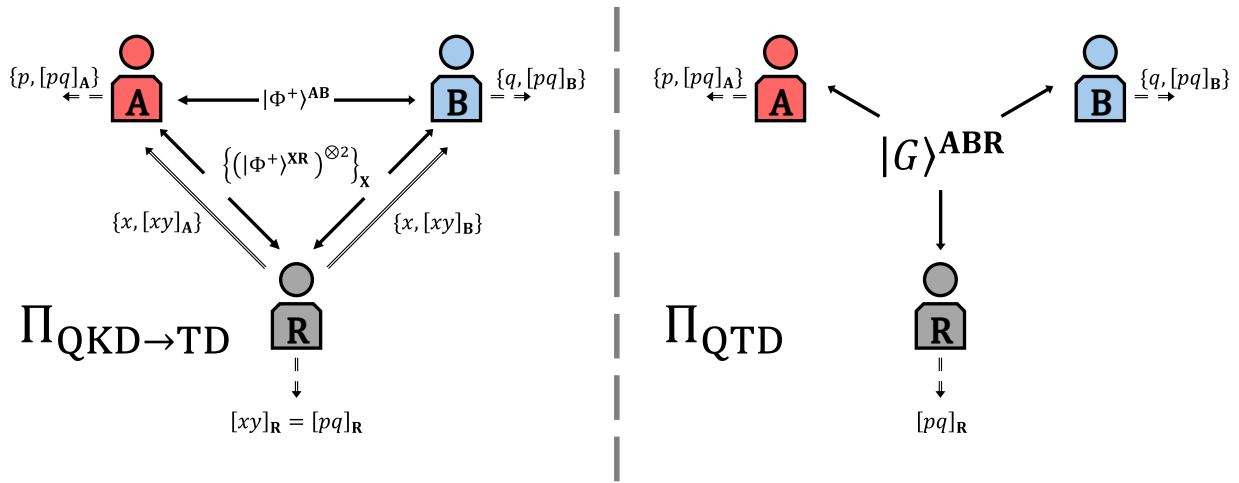}
    \caption{Network diagrams for a pair of protocols implementing $\mathcal{F}_{\text{TD}}$, which distribute private triples (where no one party knows $pq$) from entangled quantum resources. Solid single-lined arrows indicate the distribution of quantum information, while solid double-lined arrows indicate classical information sent over public channels. Information theoretic privacy in achieved in both protocols assuming an honest pair. (left) A protocol $\Pi_{\text{QKD}\to\text{TD}}$ that instantiates $\mathcal{F}_{\text{TD}}$ from QKD. Bell pairs, $\ket{\Phi^{+}}$ are distributed between parties in the network in order to establish symmetric shared secret keys. $\mathbf{R}$ first locally generates $\{x,y,xy\}$ and then uses their keys with $\mathbf{A}$ and $\mathbf{B}$ to privately send them their respective bits, across a public channel. $\mathbf{A}$ and $\mathbf{B}$, then use their own secret key to transform $\{x,y,xy\}\to\{p,q,pq\}$, which is private. (right) A protocol $\Pi_{\text{QTD}}$ that directly generates and distributes the bits in $\{p,q,pq\}$ from a single quantum resource, represented as a multipartite graph state, $\ket{G}$. No classical information is communicated in this protocol.}
    \label{fig:QTD-prob}
\end{figure}

We aim to build a protocol for disseminating private triples in such a way that we can replace $\mathbf{D}$ with a Referee $\mathbf{R}$, who both can know $f$ and need not always be trusted. With an honest majority and access to private channels there is a straightforward solution. First $\mathbf{R}$ generates a triple $\{x,y,xy\}$ locally and sends $\{x,[xy]_{\mathbf{A}}\}$ to $\mathbf{A}$ and $\{x,[xy]_{\mathbf{B}}\}$ to $\mathbf{B}$. Next, $\mathbf{A}$ and $\mathbf{B}$ share a symmetric secret key across their own private channel, which they use to transform the triple they received from $\mathbf{R}$ into one that is private, by setting  $p\leftarrow x\oplus k$, $q\leftarrow y\oplus k $, and $[pq]_{\mathbf{X}}=[xy]_{\mathbf{X}}\oplus k$, for each $\mathbf{X}\in\{\mathbf{A},\mathbf{B}\}$. If $\mathbf{R}$ chooses to retain a share $[xy]_{\mathbf{R}}$ of the original triple, they can simply set $[pq]_{\mathbf{R}}\leftarrow [xy]_{\mathbf{R}}$. Perfect privacy follows from having an honest majority. For example, the honesty of $\mathbf{B}$ and $\mathbf{R}$ ensures that $\mathbf{A}$ does not know any of the bits in the set $\{y,q,xy,pq\}$ that would violate privacy when consuming the triple in the online phase. The privacy of the value of $pq$ arise from the asymmetry that only $\mathbf{R}$ knows the value of $xy$, while only $\mathbf{A}$ and $\mathbf{B}$ know the value of their shared key $k$.

\medskip\noindent\textbf{Triple distribution from quantum resources.} The classical protocol for distributing private triples that we give above, naturally lends itself to its own protocol in the quantum domain with less assumption. This follows from the notion behind entanglement-based QKD \cite{Ekert-1991-Qc}, namely, that access to an ideal Bell pair $\ket{\Phi^{+}}=(\ket{00}+\ket{11})/\sqrt{2}$ distributed between two parties implies access to a bit of shared symmetric key. Hence, we define a protocol, $\Pi_{\text{QKD}\to\text{TD}}$, that instantiates $\mathcal{F}_{\text{TD}}$, through the use of private channels authenticated by key obtained from QKD. Assuming access to ideal quantum resources, this protocol only requires two Bell pairs shared between $\mathbf{R}$ and both $\mathbf{A}$ and $\mathbf{B}$, in order to privately disseminate the appropriate shares of $\{x,y,xy\}$, and single Bell pair shared between $\mathbf{A}$ and $\mathbf{B}$ to establish the key $k$. 

While $\Pi_{\text{QKD}\to\text{TD}}$ may be of separate practical interested, its suggests the need to generate a triple locally before extracting one that is private. Instead, we focus on proving the security of a protocol that achieves TD in a more a direct fashion. QTD defines the 3PC problem of disseminating private triples from a single quantum resource, without any form of classical communication or privately generated randomness. This is accomplished by finding the equivalent entangled quantum resource, represented as a graph state, $\ket{G}$, and measurement sequence that encodes a triple, in the same way that access to a Bell pair and local measurements yields symmetric shared key. A protocol, $\Pi_{\text{QTD}}$, that consumes a quantum resource must be more carefully designed, however, as no one party is allowed to possess a full view of the triple, unlike the case with shared key. 

In Fig. \ref{fig:QTD-prob}, we compare the network setting between $\Pi_{\text{QKD}\to\text{TD}}$ to our proposed $\Pi_{\text{QTD}}$. While we give an explicit resource state and protocol for $\Pi_{\text{QTD}}$ that requires no communication, this is based on the intuition that the resource state itself can be verified by a separate protocol. Clearly, classical communication is required to a verify such a resource, through the verification protocols discussed in Sec. \ref{sec:INT-related}. While a full numerical analysis of the communications costs to verify the resources for both protocols is left for future work, the resource state we introduce for $\Pi_{\text{QTD}}$ is of a comparable size (in terms of the number of qubits) needed for $\Pi_{\text{QKD}\to\text{TD}}$ and further removes the interaction that comes from waiting for $\mathbf{R}$'s messages in $\Pi_{\text{QKD}\to\text{TD}}$. We believe that our notion of QTD is therefore interesting from a resource theoretic perspective in cryptography, as implies the ability to extract a multiplication triple from a resource without interaction. 

\begin{figure}[]
    \centering
    \includegraphics[width=0.6\textwidth]{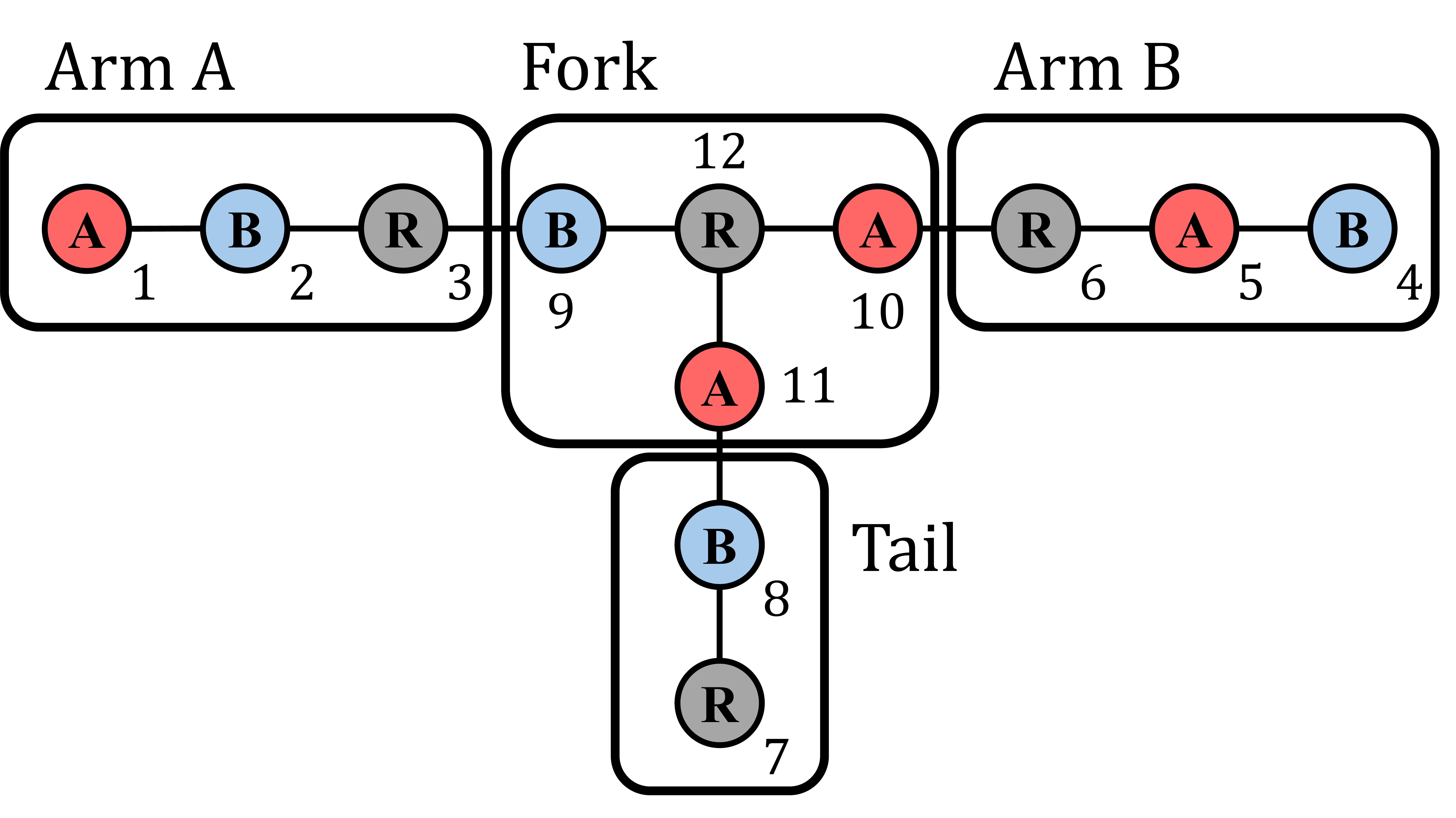}
    \caption{A graph state $\ket{G_{\wedge}}$ from which a binary multiplication triple can be extracted from the measurement-based protocol $\Pi_{\wedge}$. Each vertex represents a qubit, distributed to a party $\mathbf{X}\in\{\mathbf{A},\mathbf{B},\mathbf{R}\}$. The graph is partitioned into regions, defined by their purpose in the protocol. The numerical labeling is arbitrary. Each party is only required to measure their qubits in the Fork after those in the Appendages of the graph.}
    \label{fig:G}
\end{figure}

\section{\label{sec:OVR}A resource state for encoding triples}
We give a protocol, $\Pi_{\text{QTD}}$, for extracting a binary multiplication triple from an ideal quantum resource in the tripartite graph state $\ket{G_{\wedge}}$, depicted in Fig. \ref{fig:G}. The protocol consists of a non-interactive set of local measurements made by each party $\mathbf{X}\in\{\mathbf{A},\mathbf{B},\mathbf{R}\}$ on the resource state, labeled as its own measurement-based protocol, $\Pi_{\wedge}$, detailed in Fig. \ref{fig:Pi-QTD}. We note that the composition $\Pi_{\wedge}\circ\mathcal{F}_{\ket{G_{\wedge}}}=\Pi_{\text{QTD}}$ instantiates $\mathcal{F}_{\text{TD}}$, in the presence of $\mathcal{F}_{\ket{G_{\wedge}}}$, which distributes the ideal resource.

\medskip\noindent\textbf{The resource state}. 
The privacy achieved by the protocol is largely admitted by the structure of $\ket{G_{\wedge}}$, and who holds which qubits. For clarity, it is useful to partition the graph into regions based on their intended purpose. Qubits in $V_{\text{Arm A}}=\{1,2,3\}$, $V_{\text{Arm B}}=\{4,5,6\}$, and $V_{\text{Tail}}=\{7,8\}$ together form the Appendages of the graph. Qubits in $V_{\text{Fork}}=\{9,10,11,12\}$ form the central Fork of $\ket{G_{\wedge}}$. As it will be useful later on, we let $\Omega_{\mathbf{X}}$ be the subset of labels for the qubits $\mathbf{X}$ holds in $\ket{G_{\wedge}}$. The complete set of vertices in the graph is $V=\Omega_{\mathbf{A}}\cup\Omega_{\mathbf{B}}\cup\Omega_{\mathbf{R}}=V_{\text{Arm A}}\cup V_{\text{Arm B}}\cup V_{\text{Tail}}\cup V_{\text{Fork}}$.

\medskip\noindent\textbf{A non-interactive measurement sequence.} There is no interaction in $\Pi_{\text{QTD}}$, in the sense that no communication is necessary. Overall, the only requirement we place on the order in which qubits are measured is that each $\mathbf{X}\in\{\mathbf{A},\mathbf{B},\mathbf{R}\}$ must measure their qubits in the Appendages, before measuring their qubits in the Fork. Note that this only restricts when $\mathbf{X}$ should measure their qubit in $\Omega_{\mathbf{X}}$ with respect to each other and not the qubits of another party. This is particularly advantageous for both security and implementation. 

\begin{figure}[]
    \centering
    \framebox{
        \begin{minipage}[t]{0.9\textwidth} 
            \centerline{Protocol $\Pi_{\wedge}$}
            \medskip
            A copy of the resource $\ket{G_{\wedge}}$ is consumed by $\{\mathbf{A},\mathbf{B},\mathbf{R}\}$, who preform local quantum measurements on their qubits according to the following sequence.
            \begin{itemize}
                \item[(1)] $\mathbf{A}$, $\mathbf{B}$, and $\mathbf{R}$ measure $Z_{1}$, $X_{2}$, and $Z_{3}$, respectively, on their qubits in Arm A, obtaining measurement outcomes $m_1$, $m_2$, and $m_3$. $\mathbf{A}$ sets $p\leftarrow m_1$. 
                \item[(2)] $\mathbf{B}$, $\mathbf{A}$, and $\mathbf{R}$ measure $Z_{4}$, $X_{5}$, and $Z_{6}$, respectively, on their qubits in Arm B, obtaining measurement outcomes $m_4$, $m_5$, and $m_6$. $\mathbf{A}$ sets $q\leftarrow m_4$.
                \item[(3)] $\mathbf{R}$ and $\mathbf{B}$ measure $X_{7}$ and $Z_{8}$, respectively, on their qubits in the Tail, obtaining $m_7$ and $m_8$. $\mathbf{R}$ sets $s\leftarrow m_7$, and $\mathbf{B}$ sets $s\leftarrow m_8$.
                \item[(4)] $\mathbf{B}$ and $\mathbf{A}$ apply the unitary operations $(Z_{9})^{m_2}$ and $(Z_{11})^{m_5}$, respectively, to their qubits in the Fork.
                \item[(5)] $\mathbf{A}$, $\mathbf{B}$, and $\mathbf{R}$ make the following conditional measurements on their qubits in the Fork:
                \begin{itemize}
                    \item[(5a)] $\mathbf{B}$ measures $(W_{9})^{s}Z_{9}(W_{9}^{\dagger})^{s}$, where $W:=(iX)^{1/2}$, and obtains $m_{9}$. Note that $Y=WZW^{\dagger}$.
                    \item[(5b)] $\mathbf{A}$ measures $(W_{10})^{p}Z_{10}(W_{10}^{\dagger})^{p}$ and $(W_{11}X_{11})^{p}Z_{11}(X_{11}W_{11}^{\dagger})^{p}$, and obtains $m_{10}$ and $m_{11}$.
                    \item[(5c)] $\mathbf{R}$ measures $(U_{12})^{s}X_{12}(U_{12}^{\dagger})^{s}$, where $U:=(-iZ)^{1/2}$, and obtains $m_{12}$. Note that $Y=UXU^{\dagger}$.
                \end{itemize}
                \item[(6)] $\mathbf{A}$, $\mathbf{B}$, and $\mathbf{R}$ make following the assignments:
                \begin{align*}
                    [pq]_{\mathbf{A}}&\leftarrow m_{10}\oplus m_{11}, \\
                    [pq]_{\mathbf{B}}&\leftarrow m_{9},\\
                    [pq]_{\mathbf{R}}&\leftarrow m_{12}.
                \end{align*}
            \end{itemize}
        \end{minipage}
    }
    \caption{A measurement protocol, $\Pi_{\wedge}$, for extracting a binary multiplication triples from the resource state $\ket{G_{\wedge}}$. he combined final choice of measurement bases is private from all parties.}
    \label{fig:Pi-QTD}
\end{figure}

\medskip\noindent\textbf{Correctness}. Correctness follows from utilizing Lemma \ref{lem:G-corr} and tracing the evolution of the stabilizer of $\ket{G_{\wedge}}$ through these measurements. In steps (1)-(3) of the honest execution of $\Pi_{\text{QTD}}$, the parties perform fixed stabilizer tests on the Appendages in order to produce a set of correlations, which satisfy
\begin{subequations}
    \begin{align}
        p&\leftarrow m_1 = m_2\oplus m_3, \label{eq:p}\\ 
        q&\leftarrow m_4 = m_5\oplus m_6, \label{eq:q}\\ 
        s&\leftarrow m_7 = m_8.  \label{eq:s}
    \end{align}
\end{subequations}
This generates a uniformly random and private bit $p$ for $\mathbf{A}$ and $q$ for $\mathbf{B}$. While it may seem unnecessary that shares of $[p]$ and $[q]$ are generated between $\mathbf{R}$, and $\mathbf{B}$ or $\mathbf{A}$, accordingly, doing so allows $\mathbf{A}$ and $\mathbf{B}$ to apply the appropriate local unitaries in step (4) that encode these bits in the phase information of qubits in the central Fork, without either player learning the other's value. In contrast, the uniformly random bit $s$ is shared between $\mathbf{B}$ and $\mathbf{R}$, but private from $\mathbf{A}$ and prevents them from cheating with their additional qubit in the Fork. 

Next, in order to generate the multiplicative correlation of the triple in a way that is private, the honest execution requires that the choice of stabilizer measurement made on the Fork be conditioned on private information obtained by the parties in prior steps. Let $\Lambda(p,s)\subseteq  V_{\text{Fork}}$ be a subset that is conditioned on $p$ and $s$. In step (5), the parties collectively measure the individual components of an operator $K_{\Lambda(p,s)}$, such $(-1)^{\sum_{i\in\Lambda(p,s)}r_{i}}K_{\Lambda(p,s)}$ is an element of the stabilizer of $\ket{G_{\wedge}}$, acting only on qubits in the Fork. Table \ref{tab:C-PiQTD}, details each possible measurement combination and the corresponding correlation obtained from it. In total the correlation generated is such that
\begin{equation}
    m_{9}\oplus m_{10}\oplus m_{11}\oplus m_{12}= sr_{9}\oplus p(r_{10}\oplus r_{11})\oplus r_{12} = pq \label{eq:C-Fork}
\end{equation}
given the encoding of the phase information $\{r_{9},r_{10},r_{11},r_{12}\}=\{p,q,s,0\}$, which arises from steps (1)-(4).

\medskip\noindent\textbf{Honest-but-curious privacy.} Before proving malicious security, we demonstrate why security against an honest-but-curious adversary is trivial. Fig. \ref{fig:RealIdeal} demonstrates an abstract picture of the ideal functionality $\mathcal{F}_{\text{TD}}$ and the real protocol $\Pi_{\text{QTD}}$. In both, we give the corresponding views generated for each party, assuming honesty. In the honest execution of $\Pi_{\text{QTD}}$, we denote the view generated for $\mathbf{X}$ as $\text{view}_{\mathbf{X}}^{\Pi}:=\{m_{i}\}_{i\in\Omega_{\mathbf{X}}}$. Given that the protocol consists of only local Pauli measurements made on the ideal resource $\ket{G_{\wedge}}$, when all three parties are honest, its evident that combined views of the honest parties, $\text{view}^{\Pi}=\text{view}_{\mathbf{A}}^{\Pi}\cup \text{view}_{\mathbf{B}}^{\Pi}\cup \text{view}_{\mathbf{R}}^{\Pi}$, form three independent sets of independent and uniformly random bits. Consider $\mathbf{A}$, whose view at the end of the honest execution of $\Pi_{\text{QTD}}$ is 
\begin{equation}
    \text{view}_{\mathbf{A}}^{\Pi}=\{p,m_{5},m_{10},m_{11}\}. \label{eq:V-AH}
\end{equation}
To successfully cheat, $\mathbf{A}$ must learn $q$, however by Eq. \ref{eq:q}, $\mathbf{R}$ holds the private the measurement outcome which would reveal this information to $\mathbf{A}$. Hence, so long as everyone is honest, $P(q=0|\text{view}_{\mathbf{A}}^{\Pi})=\frac{1}{2}$. A similar argument follows for $\mathbf{B}$ and $\mathbf{R}$, and the private bits they each wish to learn. 

\medskip\noindent\textbf{The intuition behind malicious privacy.} We can go a bit further than the above, and demonstrate how an honest pair and ideal resource in $\Pi_{\text{QTD}}$ constrains the set of possible correlations a malicious adversary can obtain from $\ket{G_{\wedge}}$. Suppose that any pair parties in $\{\mathbf{A},\mathbf{B},\mathbf{R}\}$ follow the steps $\Pi_{\text{QTD}}$ honestly. Let $\mathbf{X}^{*}$ denote the remaining party, corrupted by a malicious adversary. One can show that the state held by $\mathbf{X}^{*}$, given the outcomes of the honest pair, is a stabilizer state described the stabilizer $S^{\mathbf{X}^{*}}$, such that
\begin{subequations} \label{eq:S-CX}
    \begin{align}
    S^{\mathbf{A}^{*}} &= \langle (-1)^{m_{2}\oplus m_{3}} Z_1, (-1)^{q\oplus m_{6}}X_{5}, (-1)^{m_{6}\oplus s}X_{10}X_{11}, (-1)^{(m_{2}\oplus m_{3})s\oplus m_{9}\oplus m_{12}}Z_{10}Z_{11} \rangle, \label{eq:S-CA} \\
    S^{\mathbf{B}^{*}} &= \langle (-1)^{p\oplus m_{3}}X_{2}, (-1)^{m_{5}\oplus m_{6}}Z_{4}, (-1)^{s}Z_{8},(-1)^{p(m_{5}\oplus m_{6}\oplus s)\oplus m_{3}s\oplus m_{10}\oplus m_{11}\oplus m_{12}} (W_{9})^{s}Z_{9}(W_{9}^{\dagger})^{s}\rangle, \label{eq:S-CB} \\   
    S^{\mathbf{R}^{*}} &= \langle (-1)^{p\oplus m_{2}}Z_{3}, (-1)^{q\oplus m_5}Z_{6}, (-1)^{s}X_{7}, (-1)^{pq \oplus m_{9}\oplus m_{10}\oplus m_{11}}(U_{12})^{s}X_{12}(U_{12}^{\dagger})^{s}\rangle. \label{eq:S-CR}
    \end{align}
\end{subequations}
If $\mathbf{X}^{*}\in\{\mathbf{B}^{*},\mathbf{R}^{*}\}$, the state defined by $S^{\mathbf{X}^{*}}$ is uniquely determined by $\text{view}_{\mathbf{X}}^{\Pi}$, their view in the honest execution. This is because every element of their stabilizer acts non-trivially on only a single qubit in either case, and so the actions an honest $\mathbf{X}$ would make in the protocol in fact stabilizes their system, up to the phase information they learn from $\text{view}_{\mathbf{X}}^{\Pi}$. $\mathbf{A}$, on the other hand, does not learn the complete set of bits that parametrizes $S^{\mathbf{A}^*}$, from their honest execution of the protocol. One way to see this is that while $\text{view}_{\mathbf{A}}^{\Pi}$ contains four bits, it's the sum of bits $m_{10}\oplus m_{11}$ that is actually correlated with $pq$. In order to fully determine their stabilizer, a corrupt $\mathbf{A}^*$ simply can choose to follow $\Pi_{\text{QTD}}$ through step (4), and then jointly measure $X_{10}X_{11}$ in order obtain to a bit $q\oplus s$, without disturbing the rest of the honest execution. Therefore, we combine $m_{10}\oplus m_{11}$ into a single bit and append the bit $q\oplus s$ to $\text{view}_{\mathbf{A}}^{\Pi}$ in Eq. \ref{eq:V-AH}, such that 
\begin{equation}
    \text{view}_{\mathbf{A}}^{\Pi}=\{p,m_{5},m_{10}\oplus m_{11},q\oplus s\}, \label{eq:V-AH-mod}
\end{equation}
which captures the information a corrupt $\mathbf{A}^*$ needs to fully characterize their state, up to some local transformations of the bits. Note that $P(q=0|q\oplus s)=1/2$, and so honest-but-curious privacy remains even if this bit were obtained by $\mathbf{A}$ in the honest execution of $\Pi_{\text{QTD}}$.

\begin{table}[]
    \centering
    \setlength{\tabcolsep}{10pt} 
    \renewcommand{\arraystretch}{1.5} 
    \begin{tabular}{ccc}
        \toprule
        \multicolumn{1}{l}{$\Lambda(p,s)$} & Local measurement bases & Correlation\\
        \hline
        \multicolumn{1}{l}{$\Lambda(0,0)=\{12\}$} & $Z_{9}^{\mathbf{B}},Z_{10}^{\mathbf{A}},Z_{11}^{\mathbf{A}},X_{12}^{\mathbf{R}}$ & \multicolumn{1}{l}{$m_{9}\oplus m_{10}\oplus m_{11}\oplus m_{12} = 0$}\\
        \multicolumn{1}{l}{$\Lambda(0,1)=\{9,12\}$} & $Y_{9}^{\mathbf{B}},Z_{10}^{\mathbf{A}},Z_{11}^{\mathbf{A}},Y_{12}^{\mathbf{R}}$ & \multicolumn{1}{l}{$m_{9}\oplus m_{10}\oplus m_{11}\oplus m_{12} = 0$} \\
        \multicolumn{1}{l}{$\Lambda(1,0)=\{10,11,12\}$} & $Z_{9}^{\mathbf{B}},Y_{10}^{\mathbf{A}},-Y_{11}^{\mathbf{A}},X_{12}^{\mathbf{R}}$ & \multicolumn{1}{l}{$m_{9}\oplus m_{10}\oplus m_{11}\oplus m_{12} = q$}\\
        \multicolumn{1}{l}{$\Lambda(1,1)=\{9,10,11,12\}$} & $Y_{9}^{\mathbf{B}},Y_{10}^{\mathbf{A}},-Y_{11}^{\mathbf{A}},Y_{12}^{\mathbf{R}}$ & \multicolumn{1}{l}{$m_{9}\oplus m_{10}\oplus m_{11}\oplus m_{12} = q$} \\
    \bottomrule
    \end{tabular}
    \caption{Table of correlations obtained at the end of $\Pi_{\wedge}$, from the conditional stabilizer test made on the Fork of $\ket{G_{\wedge}}$.}
    \label{tab:C-PiQTD}
\end{table}

The intuition behind security is that the honest pair ensures that $\mathbf{X}^*$ only has access to a pure state, $\ket{\text{view}^{\Pi}_{\mathbf{X}}}$, that is uniquely characterized by the honest view of $\mathbf{X}$, in which every sensitive bit value the adversary might want to learn is padded by some uniformly random bit. While in reality this has more to do with the nature of one-time pads than anything physical, we can state this fact precisely as follows. The stabilizers expressed in Eq. \ref{eq:S-CX} can be written in a general form, as $S^{\mathbf{X}^{*}}=\langle\{(-1)^{c_i}g_{i}\}_{i\in\Omega_{\mathbf{X}}}\rangle$, where each $g_i$ is one of the four appropriate generators, such that $\text{view}_{\text{X}}^{\Pi}=\{c_{i}\}_{i\in\Omega_{\mathbf{X}}}$. Note, for the case of $\mathbf{A}^{*}$, we are making an association between $g_{10}=X_{10}X_{11}$ and $g_{11}=Z_{10}Z_{11}$ to keep notation consistent, and their $\text{view}_{\mathbf{A}}^{\Pi}$ comes from a set of transformations made on $\{c_{i}\}_{i\in\Omega_{\mathbf{X}}}$, rather than direct equality. From this, the corrupted party can describe their local quantum system as
\begin{equation}
    \ketbra{\text{view}_{\mathbf{X}}^{\Pi}}:=\frac{1}{2^{4}}\prod_{i\in\Omega_{\mathbf{X}}}\left(\mathbb{I}+(-1)^{c_{i}}g_{i}\right), \label{eq:psi-X}
\end{equation}
What we care about is the ability of $\mathbf{X}^{*}$ to guess the values of a certain subset of sensitive bits, $\Xi_{\mathbf{X}}\subset\text{view}^{\Pi}\backslash\text{view}_{\mathbf{X}}^{\Pi}$, that would reveal the value of $pq$. Explicitly, $\Xi_{\mathbf{A}}=\{q\}$, $\Xi_{\mathbf{B}}=\{p\}$, and $\Xi_{\mathbf{R}}=\{p,q\}$. For clarity, let $\Upsilon_{\mathbf{X}}=\text{view}^{\Pi}\backslash\text{view}_{\mathbf{X}}^{\Pi}\cup\Xi_{\mathbf{X}}$ be the set of bits which on their own provide no information about the triple. Prior to any attack that $\mathbf{X}^{*}$ might employ, they can express their state of knowledge of $\Xi_{\mathbf{X}}$ as a classical-quantum (cq) state
\begin{equation}
    \rho_{\text{cq}}^{\mathbf{X}^{*}}(\Xi_{\mathbf{X}},\text{view}_{\mathbf{X}}^{\Pi})=\frac{1}{2^{|\Upsilon_{\mathbf{X}}|}}\sum_{\Upsilon_{\mathbf{X}}}\ketbra{\Xi_{\mathbf{X}},\Upsilon_{\mathbf{X}}}^{\mathsf{X}}\otimes\ketbra{\text{view}_{\mathbf{X}}^{\Pi}}, \label{eq:rhocq-X}
\end{equation}
where the superscript $\mathsf{X}$ denotes the set of classical registers, and the sum is over all possible measurement outcomes in $\Upsilon_{\mathbf{X}}$. It is a straightforward calculation from Eq. \ref{eq:rhocq-X} that $\Tr_{\mathsf{X}}\left[\rho_{\text{cq}}^{\mathbf{A}^*}(p;\text{view}_{\mathbf{A}}^{\Pi})\right]=\Tr_{\mathsf{X}}\left[\rho_{\text{cq}}^{\mathbf{B}^{*}}(q;\text{view}_{\mathbf{B}}^{\Pi})\right]=\Tr_{\mathsf{X}}\left[\rho_{\text{cq}}^{\mathbf{R}^*}(p,q;\text{view}_{\mathbf{R}}^{\Pi})\right]=\mathbb{I}/2^{4}$. In more detail, $\mathbf{X}^{*}$'s description of their physical system, marginalized over the bits in $\Pi_{\text{QTD}}$ that provide no information outside of the protocol, is maximally mixed. The view then generated from $\ket{\text{view}^{\Pi}_{\mathbf{X}}}$ is therefore uncorrelated from $\Xi_{\mathbf{X}}$, in the sense 
\begin{equation}
    P\left(x=0\middle|\rho_{\text{cq}}^{\mathbf{X}^{*}}(\Xi_{\mathbf{X}},\text{view}_{\mathbf{X}}^{\Pi})\right)=1/2,
\end{equation}
for each $x\in\Xi_{\mathbf{X}}$. In the next section, we formalize our notion of security from this fact.

\begin{figure}
    \centering
    \includegraphics[width=0.99\textwidth]{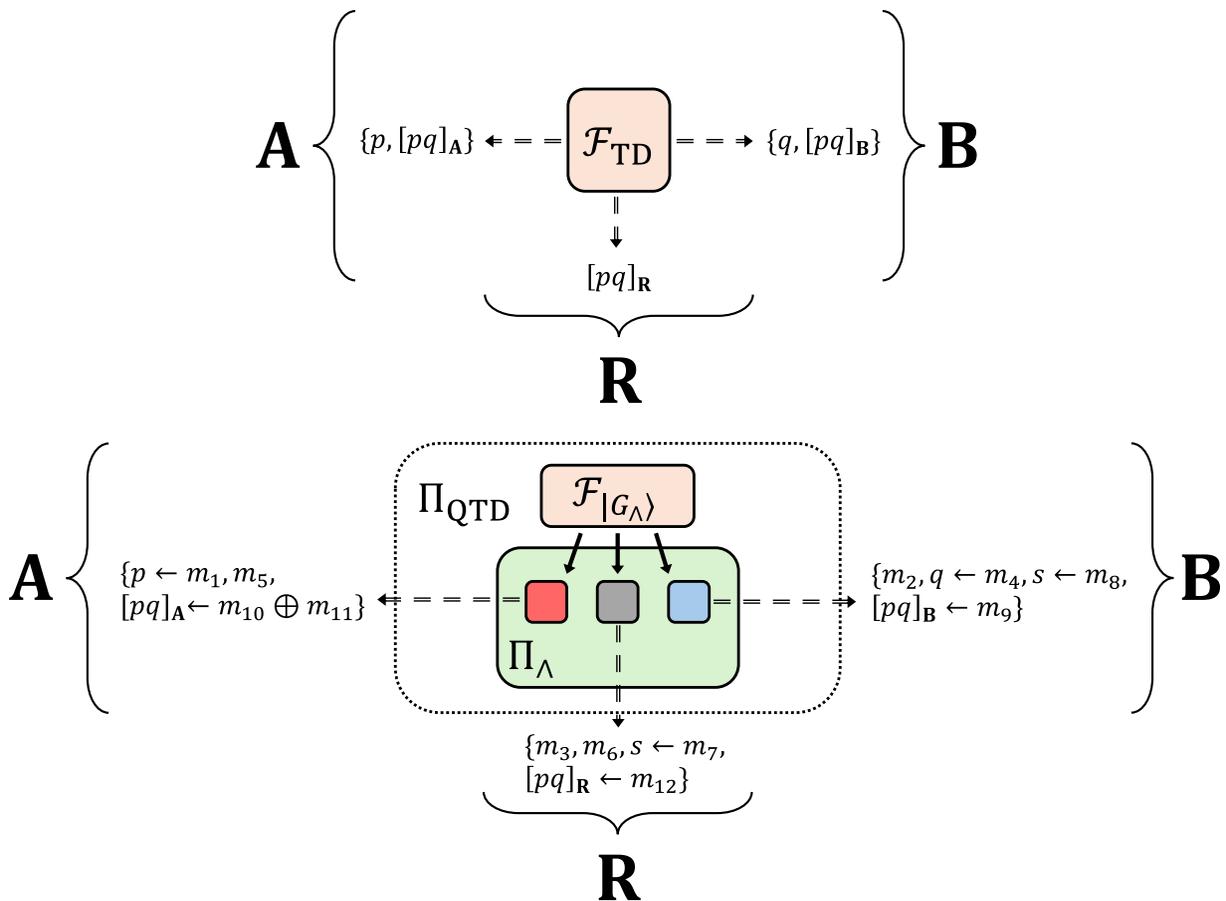}
    \caption{A high-level picture of an ideal functionality $\mathcal{F}_{\text{TD}}$ (top) for secret sharing a binary multiplication triple among three parties, and the real protocol, $\Pi_{\text{QTD}}$ (bottom), which securely realizes it. Here we depict the view generated for each party by the honest execution of the protocol. Solid single-lined arrows indicate quantum information distributed to each party, while dashed double-lined arrows indicate private classical information.}
    \label{fig:RealIdeal}
\end{figure}

\section{\label{sec:SEC}Malicious-secure privacy}
We prove perfect privacy for the shares obtained from $\Pi_{\text{QTD}}$, against a malicious adversary who corrupts only a single $\mathbf{X}^{*}$. As there is no interaction in $\Pi_{\text{QTD}}$, and we assume here that a verified copy of $\ket{G_{\wedge}}$ has been distributed, the only classical information $\mathbf{X}^{*}$ can extract from the protocol is from their local quantum system. Given an honest pair, the state $\ket{\text{view}^{\Pi}_{\mathbf{X}}}$, which the corrupted $\mathbf{X}^{*}$ can act on in $\Pi_{\text{QTD}}$, is a unique pure state determined by a set of independent and uniformly random bits, $\text{view}^{\Pi}_{\mathbf{X}}$, that $\mathbf{X}^{*}$ would otherwise learn in the honest execution. An attack then executed by $\mathbf{X}^{*}$ could be described by some arbitrary local action on their quantum system, however, in the previous section, we demonstrated that the only classical information that can be extracted from their system is padded, leaking no information. Therefore, we can formalize this guarantee of privacy through the construction of a simple simulator, $\mathbb{S}$, that interacts with $\mathcal{F}_{\text{TD}}$ and reconstructs a $\rho_{\text{cq}}^{\mathbf{X}^{*}}(\Xi_{\mathbf{X}})$ for the adversary that yields a perfectly indistinguishable view from the one they would a obtain when attacking the real protocol. For each possible role in $\Pi_{\textup{QTD}}$, Fig. \ref{fig:Sim} depicts a simulator that replicates the view generated for the adversary in the real protocol.

\begin{theorem} \label{thm:Sec}
The protocol $\Pi_{\textup{QTD}}$ realizes $\mathcal{F}_{\textup{TD}}$ with perfect privacy, in the presence of a functionality $\mathcal{F}_{\ket{G_{\wedge}}}$ that securely distributes $\ket{G_{\wedge}}$, given a malicious adversary that corrupts only a single $\mathbf{X}\in\{\mathbf{A},\mathbf{B},\mathbf{R}\}$, labeled $\mathbf{X}^{*}$. Explicitly, given an honest pair in $\Pi_{\textup{QTD}}$, there exists a simulator $\mathbb{S}$ interacting with $\mathcal{F}_{\textup{TD}}$, such that the cq state held by $\mathbf{X}^{*}$, following their attack on the real or ideal protocol, satisfies
\begin{equation}
    P\left(x=0\middle|\rho_{\textup{cq}}^{\mathbf{X}^{*}}(\Xi_{\mathbf{X}},\textup{view}_{\mathbf{X}}^{\Pi})\right)-P\left(x=0\middle|\rho_{\textup{cq}}^{\mathbf{X}^{*}}(\Xi_{\mathbf{X}},\textup{view}^{\mathcal{F}}_{\mathbf{X}},\mathbb{S})\right)=0, \label{eq:S-Ind}
\end{equation}
for each $x\in\Xi_{\mathbf{R}}$, and every possible $\textup{view}_{\mathbf{X}}^{\Pi}$ generated for the adversary in the real protocol.

\end{theorem}
\begin{proof}
$\mathbb{S}$ first calls $\mathcal{F}_{\text{TD}}$, in order to obtain $\text{view}_{\mathbf{X}}^{\mathcal{F}}$. It then generates an appropriate number of independent and uniformly random bits that it appends to $\text{view}_{\mathbf{X}}^{\mathcal{F}}$, obtaining $\text{view}_{\mathbf{X}}^{\Pi}$ in the process.  Next, it builds the state $\ket{\text{view}_{\mathbf{X}}^{\Pi}}$, according to stabilizers given in Eq. \ref{eq:S-CX}, and sends it $\mathbf{X}^{*}$. We can work through an explicit example, in the case when $\mathbf{X}^*=\mathbf{A}^{*}$. $\mathbb{S}$ calls $\mathcal{F}_{\text{TD}}$ and receives uniformly random bits $p,[pq]_{\mathbf{A}}$. In order to generate a state that is consistent with the output of the real protocol, the simulator must construct a set of bits $\{c_{1},c_{5},c_{10},c_{11}\}$ that captures $\text{view}_{\mathbf{A}}^{\Pi}$. It stars by generating a pair of independent and uniformly random bits, say $x$ and $y$, and proceeds with the assignments
\begin{subequations}
    \begin{align}
        c_{1} &\leftarrow p, \\
        c_{5} &\leftarrow x, \\
        c_{10} &\leftarrow y, \\
        c_{11} &\leftarrow [pq]_{\mathbf{A}}\oplus py.
    \end{align}
\end{subequations}
$\mathbb{S}$ can then constructs the corresponding $\ket{\text{view}_{\mathbf{A}}^{\Pi}}$, according to stabilizer in Eq. \ref{eq:S-CA}, which it then outputs to the adversary. The example is even simpler for the cases, when $\mathbf{X}^*\in\{\mathbf{B^{*},\mathbf{R}^*}\}$, as $\mathbb{S}$ does not need to perform any local bit operations between the output it receives from $\mathcal{F}_{\text{TD}}$ and the random bits it generates.

In total, $\mathbb{S}$ generates a cq state $\rho_{cq}^{\mathbf{X}^*}(\Xi_{\mathbf{X}};\text{view}_{\mathbf{X}}^{\mathcal{F}},\mathbb{S})$ for adversary, describing their state of knowledge for any bits in $\Xi_{\mathbf{X}}$. It's trivial this state provides a perfectly indistinguishable state of knowledge of $\Xi_{\mathbf{X}}$ from the one obtained by adversary in their attack on $\Pi_{\text{QTD}}$, as here it is explicitly independent from those bits. To see that $\rho_{cq}^{\mathbf{X}^*}(\Xi_{\mathbf{X}};\text{view}_{\mathbf{X}}^{\mathcal{F}},\mathbb{S})$ is consistent with the output the adversary receives from this attack, it is straightforward to show how $\mathbf{X}^*$ can obtain $\text{view}_{\mathbf{X}}^{\Pi}$, without disturbing their state, by making appropriate stabilizer measurements.
\end{proof}

\begin{figure}
    \centering
    \includegraphics[width=0.9\textwidth]{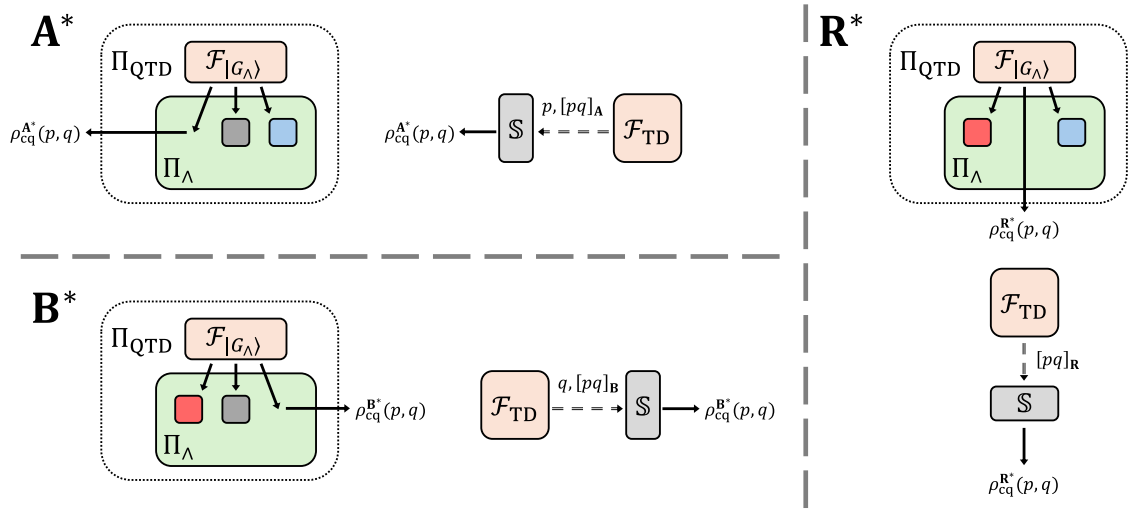}
    \caption{Diagrammatic representations of malicious attacks executed on $\Pi_{\text{QTD}}$, by a singular corrupt $\mathbf{X}^{*}\in\{\mathbf{A}^{*},\mathbf{B}^{*}, \mathbf{R}^{*}\}$. In each case, we demonstrate that the adversary receives from the real execution a cq state of the bits in the triple that is equivalent to one obtained when interacting with the corresponding simulation $\mathbb{S}$, which emulates the attack in the ideal execution.}
    \label{fig:Sim}
\end{figure}

Indeed the construction of the simulator is straightforward given our setup assumptions. Nothing is input into the protocol, so there is no need for $\mathbb{S}$ to extract any information from the adversary and send to the ideal functionality. Furthermore, the protocol is non-interactive to the extent that no messages are sent whatsoever. A result that directly follows from such a simulation-based proof, is that perfect privacy holds through the sequential composition of $\Pi_{\text{QTD}}$ \cite{Beaver-1991-Sm}. However, as $\Pi_{\text{QTD}}$ is non-interactive and input-independent, this trivially extends to parallel and concurrent composition as well under the same definition. Hence, we can construct a completely private offline phase for some MPC in the pre-processing model (i.e. one that utilizes proven information theoretic primitives in the online phase), which only relies on $\Pi_{\text{QTD}}$ in order to distribute the required set of computational resources.


\section{\label{sec:APP}Applications}
The ability to extract triples from an entangled quantum resource with information-theoretic privacy motivates other protocols built on this assumption that can be executed with the same guarantees. The main example of this is secret-sharing-based MPC, within the pre-processing model. In this setting, computation is divided into two parts: 1.) an offline phase, where triples are distributed, and 2.) an online phase in which parties utilize a broadcast channel to open padded messages and perform local computation on their shares. Let $f$ be some function for which a set of parties provide inputs. The goal is to compute $f$, such that none of parties learn anything about one another's input, other than what is revealed by $f$ itself. We will separate the final reconstruction step of the online phase, where parties make the last set of broadcasts revealing their shares of $[f]_{\mathbf{X}}$, for the sake of isolating when an adversary should be able to learn anything about the inputs of $f$. For simplicity, we will assume an ideal public broadcast channel that is only subject to a passive eavesdropper, in the sense that anyone can listen but an adversary is unable to modify messages sent.

From Theorem. \ref{thm:Sec}, we can build a secure offline phase from only calls to $\Pi_{\text{QTD}}$, disseminating some number of triples. In \cite{Gold-2025-Hp}, we utilized a variant of this protocol as a primitive to build up an MPC for evaluating Boolean functions of a limited class. A Boolean $f$ can be expressed in its algebraic normal form (ANF), where we write $f=\bigoplus_{i}\mathfrak{c}_{i}$, with each $\mathfrak{c}_{i}$ being some product of party inputs. The MPC from \cite{Gold-2025-Hp} admitted any $f$ of this form, where each $\mathfrak{c}_{i}$ is either a quadratic conjunction of distinct parties' inputs, $x_{i}y_{i}$, or a linear bit $z_{i}$ belonging to a single party. This supports any two-party function (one taking inputs from two unique parties) in our 3PC setting, in the sense that $\mathbf{R}$ can assist $\mathbf{A}$ and $\mathbf{B}$ in computing any Boolean function of $\mathbf{A}$'s and $\mathbf{B}$'s inputs. As a follow-up, we show here how to build 1-out-of-2 OT in this setting, where $\mathbf{A}$ plays the role of the Sender and $\mathbf{B}$ the Receiver. $\mathbf{R}$ assists in the computation, by keep their shares private until the final reconstruction step. We expand these results, to show how any $N$-party Boolean function can be evaluated with perfect privacy, by demonstrating how to construct a sharing of a bit conjunction at arbitrary order. In both examples below, we focus on building a $\mathcal{F}$ that calls only $\mathcal{F}_{\text{TD}}$ and the ideal public broadcast channel. One can then construct the corresponding $\Pi$, by simply replacing each call to $\mathcal{F}_{\text{TD}}$ with $\Pi_{\text{QTD}}$.

\medskip\noindent\textbf{Referee-assisted 1-out-of-2 OT.} The ability to instantiate a functionality for OT from multiplication triples should be unsurprising. Their are known examples for distributing triples that function in the other direction, utilizing OT \cite{Frederiksen-2016-UA, Keller-2016-MF}. Without any additional functionalities in the offline phase, a functionality for 1-out-of-2 OT, $\mathcal{F}_{\text{OT}}$, can be constructed from only $\mathcal{F}_{\text{TD}}$ and a public broadcast channel, depicted in Fig. \ref{fig:F-OT}. Here $\mathbf{A}$ plays the role of the Sender, choosing a pair of input bits $a_{0},a_{1}$, and $\mathbf{B}$ plays the Receiver, choosing a single bit $b$. They work together with $\mathbf{R}$ to obtain a pair of triples, from which they help $\mathbf{B}$ securely compute
\begin{equation}
    f(a_{0},a_{1},b)=a_{0}(b\oplus1)\oplus a_{1}b, \label{eq:OT}
\end{equation}
in the online the phase. In order to keep the result $\mathbf{B}$ gets private from $\mathbf{A}$ and $\mathbf{R}$, only $\mathbf{A}$ and $\mathbf{R}$ broadcast messages during the reconstruction step.

\begin{figure}[]
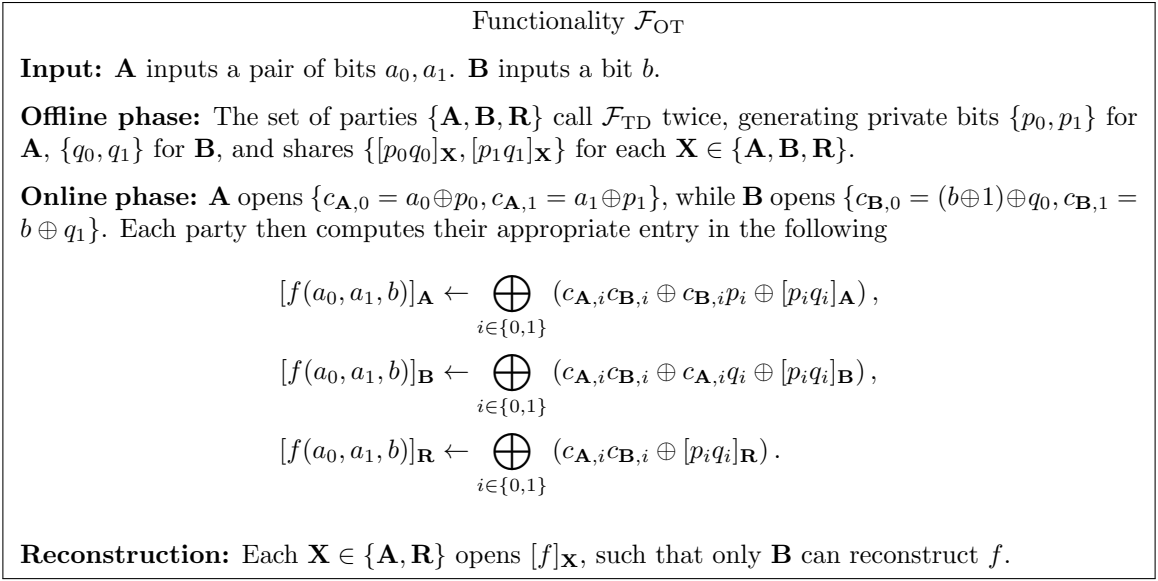

    \centering
    \framebox{
        \begin{minipage}[t]{0.9\textwidth} 
            \centerline{Functionality $\mathcal{F}_{\text{OT}}$}
            \medskip
            \textbf{Input:} $\mathbf{A}$ inputs a pair of bits $a_{0},a_{1}$. $\mathbf{B}$ inputs a bit $b$.
            
            \medskip
            \textbf{Offline phase:} The set of parties $\{\mathbf{A},\mathbf{B},\mathbf{R}\}$ call $\mathcal{F}_{\text{TD}}$ twice, generating private bits $\{p_{0},p_{1}\}$ for $\mathbf{A}$, $\{q_{0},q_{1}\}$ for $\mathbf{B}$, and shares $\{[p_{0}q_{0}]_{\mathbf{X}},[p_{1}q_{1}]_{\mathbf{X}}\}$ for each $\mathbf{X}\in\{\mathbf{A},\mathbf{B},\mathbf{R}\}$.

            \medskip
            \textbf{Online phase:} $\mathbf{A}$ opens $\{c_{\mathbf{A},0}=a_{0}\oplus p_{0},c_{\mathbf{A},1} =a_{1}\oplus p_{1}\}$, while $\mathbf{B}$ opens $\{c_{\mathbf{B},0}=(b\oplus 1)\oplus q_{0},c_{\mathbf{B},1} =b\oplus q_{1}\}$. Each party then computes their appropriate entry in the following
            \begin{align*}
                [f(a_0,a_1,b)]_{\mathbf{A}}&\leftarrow\bigoplus_{i\in\{0,1\}} \left(c_{\mathbf{A},i}c_{\mathbf{B},i}\oplus c_{\mathbf{B},i}p_{i}\oplus[p_{i}q_{i}]_{\mathbf{A}} \right), \\
                [f(a_0,a_1,b)]_{\mathbf{B}}&\leftarrow\bigoplus_{i\in\{0,1\}} \left(c_{\mathbf{A},i}c_{\mathbf{B},i}\oplus c_{\mathbf{A},i}q_{i}\oplus[p_{i}q_{i}]_{\mathbf{B}} \right), \\
                [f(a_0,a_1,b)]_{\mathbf{R}}&\leftarrow\bigoplus_{i\in\{0,1\}} \left(c_{\mathbf{A},i}c_{\mathbf{B},i}\oplus [p_{i}q_{i}]_{\mathbf{R}} \right). 
            \end{align*}
            
            \medskip
            \textbf{Reconstruction:} Each $\mathbf{X}\in\{\mathbf{A},\mathbf{R}\}$ opens $[f]_{\mathbf{X}}$, such that only $\mathbf{B}$ can reconstruct $f$.
        \end{minipage}
    }
    \caption{A functionality for implementing 1-out-of-2 OT, utilizing only $\mathcal{F}_{\text{TD}}$ and broadcast messaging.}
    \label{fig:F-OT}
\end{figure}

\medskip\noindent\textbf{Boolean MPC.} 
Suppose that $N$-parties, $\{\mathbf{P}_{1},\cdots,\mathbf{P}_{N}\}$, which to evaluate some conjunction of bits
\begin{equation}
    f^{(N)}(x_{1},\cdots,x_{N})=x_{1}\cdots x_{N}=\prod_{k\in\{1,\cdots,N\}}x_{k},
\end{equation}
where each party $\mathbf{P}_{k}$ supplies an input $x_{k}$. Clearly, when $N=2$, the pair, $\{\mathbf{P}_{1},\mathbf{P}_{2}\}$, can work together with a Referee, $\mathbf{R}$, to generate a single triple of the form given by $\mathcal{F}_{\text{TD}}$, from which they can achieve an additive sharing of the quadratic bit conjunction $[x_{1}x_{2}]$ between $\{\mathbf{P}_{1},\mathbf{P}_{2},\mathbf{R}\}$ with only a single round of openings. As we proved above, this can be implemented with perfect privacy with our protocol, given that any pair in $\{\mathbf{P}_{1},\mathbf{P}_{2},\mathbf{R}\}$ is honest, and an ideal resource state is consumed. In the general case, where we have a set of parties, $\{\mathbf{P}_{1},\cdots,\mathbf{P}_{N},\mathbf{R}\}$, we can build up higher order conjunctions using the inductive argument that follows. 

Suppose the set $\{\mathbf{P}_{1},\cdots,\mathbf{P}_{N-1},\mathbf{R}\}$ already holds an additive sharing of $[x_{1}\cdots x_{N-1}]$, which they wish to promote to the $N^{\text{th}}$ order conjunction by working with a party $\mathbf{P}_{N}$, holding a private $x_{N}$. First, $\mathbf{R}$ opens theirs share $[x_{1}\cdots x_{N-1}]_{\mathbf{R}}$, and, without loss of generality, we assume $\mathbf{P}_{N-1}$ absorbs this bit into their share $[x_{1}\cdots x_{N-1}]_{\mathbf{P}_{N-1}}$. Note that having $\mathbf{R}$ make this additional opening is not a problem, as ideally this bit is independent and uniformly random, and therefore reveals nothing about the secret being shared. The goal of this step is simply to have $\mathbf{R}$ remove themselves from the sharing, which if necessary can also be achieved by calling on some additional secret sharing functionality as well. Now that the sharing $[x_{1}\cdots x_{N-1}]$ is held by only $N-1$ parties, we consume $N-1$ triples, each held between a subset $\{\mathbf{P}_{k},\mathbf{P}_{N},\mathbf{R}\}$ for each $k\in\{1,\cdots,N\}$, such that each possible subset can generate their own sharing of $[[x_1\cdots x_{N-1}]_{\mathbf{P_{k}}}x_{N}]$. From here it follows
\begin{equation}
    x_{1}\cdots x_{N} = \bigoplus_{k\in\{1,\cdots,N-1\}} \left([[x_{1}\cdots x_{N-1}]_{\mathbf{P}_{k}}x_{N}]_{\mathbf{P}_{k}} \oplus[[x_{1}\cdots x_{N-1}]_{\mathbf{P}_{k}}x_{N}]_{\mathbf{P}_{N}} \oplus[[x_{1}\cdots x_{N-1}]_{\mathbf{P}_{k}}x_{N}]_{\mathbf{R}} \right). \label{eq:Nth-eval}
\end{equation}
In total, generating a sharing of an $N^{\text{th}}$-order bit conjunction requires a number of triples, $\sum_{k\in\{1,\cdots,N-1\}}k=N(N-1)/2=\binom{N}{2}\leq N^{2}/2$. In the offline phase, each possible combination of parties, $\{\mathbf{P}_{k},\mathbf{P}_{l}\}$, work with $\mathbf{R}$ to call $\mathcal{F}_{\text{TD}}$ (or $\Pi_{\text{QTD}}$ in the real instantiation). The round complexity required here in the online phase is $N-1$, as all of the openings required to locally evaluate Eq. \ref{eq:Nth-eval} can happen concurrently, which holds when building up each higher order bit conjunction. It is similarly simple to have $\mathbf{R}$ open each $[pq]_{\mathbf{R}}$ they obtain from each call $\mathcal{F}_{\text{TD}}$, at the same time prior to any other broadcasts in the online phase. This adds only a single round to the overall complexity. As an explicit example, Fig. \ref{fig:F-f3} demonstrates one way that parties $\{\mathbf{P}_{1},\mathbf{P}_{2},\mathbf{P}_{3},\mathbf{R}\}$ can work together to privately evaluate a quartic bit conjunction, $f^{(3)}(x_{1},x_{2},x_{3})=x_{1}x_{2}x_{3}$, through a functionality, $\mathcal{F}_{f^{(3)}}$ that calls only $\mathcal{F}_{\text{TD}}$ and the ideal broadcast channel.

When calling on the real $\Pi_{\text{QTD}}$ to distribute triples, the simplest way to maintain a concrete understanding of privacy in this multi-party setting, is to require that the $\mathbf{R}$ be among the subset of honest parties. Given an honest $\mathbf{R}$, then any $\mathbf{P}_{k}$ can guarantee the privacy of their inputs by playing honestly themselves. Therefore, while $\Pi_{\text{QTD}}$ requires an honest majority within any subset $\{\mathbf{P}_{k},\mathbf{P}_{l},\mathbf{R}\}$ of the full set of parties, when $N>3$, a composition of only calls to our protocol and an ideal broadcasting channel allows for the evaluation of any Boolean function with an honest minority. Furthermore, in each call to $\Pi_{\text{QTD}}$, an honest $\mathbf{R}$'s opening of their share $[pq]_{\mathbf{R}}$ has no effect on the ability of a corrupt $\mathbf{A}^{*}$ or $\mathbf{B}^{*}$ to discern any sensitive information. This is ensured by the structure of the phase information in Eqs. \ref{eq:S-CA} and \ref{eq:S-CB}, where the elements of the other honest party's view remains as a pad.

\begin{figure}[]
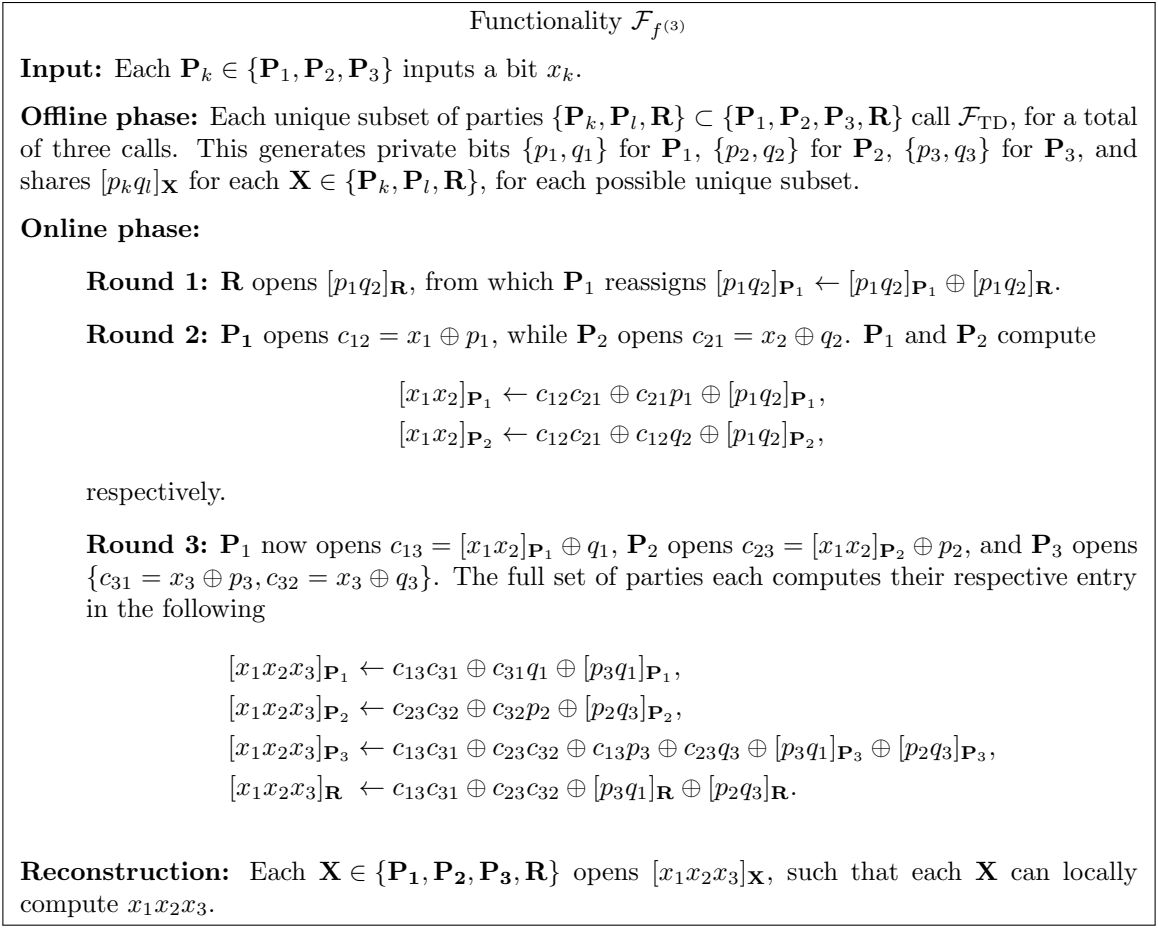

    \centering
    \framebox{
        \begin{minipage}[t]{0.9\textwidth} 
            \centerline{Functionality $\mathcal{F}_{f^{(3)}}$}
            \medskip
            \textbf{Input:} Each $\mathbf{P}_{k}\in\{\mathbf{P}_{1},\mathbf{P}_{2},\mathbf{P}_{3}\}$ inputs a bit $x_{k}$.
            
            \medskip
            \textbf{Offline phase:} Each unique subset of parties $\{\mathbf{P}_{k},\mathbf{P}_{l},\mathbf{R}\}\subset\{\mathbf{P}_{1},\mathbf{P}_{2},\mathbf{P}_{3},\mathbf{R}\}$ call $\mathcal{F}_{\text{TD}}$, for a total of three calls. This generates private bits $\{p_{1},q_{1}\}$ for $\mathbf{P}_{1}$, $\{p_{2},q_{2}\}$ for $\mathbf{P}_{2}$, $\{p_{3},q_{3}\}$ for $\mathbf{P}_{3}$, and shares $[p_{k}q_{l}]_{\mathbf{X}}$ for each $\mathbf{X}\in\{\mathbf{P}_{k},\mathbf{P}_{l},\mathbf{R}\}$, for each possible unique subset.

            \medskip
            \textbf{Online phase:} 
            \begin{itemize}
                \item[]\textbf{Round 1:} $\mathbf{R}$ opens $[p_{1}q_{2}]_{\mathbf{R}}$, from which $\mathbf{P}_{1}$ reassigns $[p_{1}q_{2}]_{\mathbf{P}_{1}}\leftarrow[p_{1}q_{2}]_{\mathbf{P}_{1}}\oplus [p_{1}q_{2}]_{\mathbf{R}}$.
                \item[]\textbf{Round 2:} $\mathbf{P_{1}}$ opens $c_{12}=x_{1}\oplus p_{1}$, while $\mathbf{P}_{2}$ opens $c_{21}=x_{2}\oplus q_{2}$. $\mathbf{P}_{1}$ and $\mathbf{P}_{2}$ compute
                \begin{align*}
                    [x_{1}x_{2}]_{\mathbf{P}_{1}}&\leftarrow c_{12}c_{21}\oplus  c_{21}p_{1}\oplus [p_{1}q_{2}]_{\mathbf{P}_{1}}, \\
                    [x_{1}x_{2}]_{\mathbf{P}_{2}}&\leftarrow c_{12}c_{21}\oplus  c_{12}q_{2}\oplus [p_{1}q_{2}]_{\mathbf{P}_{2}},
                \end{align*}
                respectively.
                \item[]\textbf{Round 3:} $\mathbf{P}_{1}$ now opens $c_{13}=[x_{1}x_{2}]_{\mathbf{P}_{1}}\oplus q_{1}$, $\mathbf{P}_{2}$ opens $c_{23}=[x_{1}x_{2}]_{\mathbf{P}_{2}}\oplus p_{2}$, and $\mathbf{P}_{3}$ opens $\{c_{31}=x_{3}\oplus p_{3},c_{32}=x_{3}\oplus q_{3}\}$. The full set of parties each computes their respective entry in the following
                \begin{align*}
                    [x_{1}x_{2}x_{3}]_{\mathbf{P}_{1}} &\leftarrow c_{13}c_{31}\oplus c_{31}q_{1}\oplus[p_{3}q_{1}]_{\mathbf{P}_{1}}, \\
                    [x_{1}x_{2}x_{3}]_{\mathbf{P}_{2}} &\leftarrow c_{23}c_{32}\oplus c_{32}p_{2}\oplus[p_{2}q_{3}]_{\mathbf{P}_{2}}, \\
                    [x_{1}x_{2}x_{3}]_{\mathbf{P}_{3}} &\leftarrow c_{13}c_{31}\oplus  c_{23}c_{32}\oplus c_{13}p_{3}\oplus c_{23}q_{3}\oplus [p_{3}q_{1}]_{\mathbf{P}_{3}}\oplus[p_{2}q_{3}]_{\mathbf{P}_{3}}, \\
                    [x_{1}x_{2}x_{3}]_{\mathbf{R}_{\,\,}} &\leftarrow c_{13}c_{31}\oplus  c_{23}c_{32}\oplus [p_{3}q_{1}]_{\mathbf{R}}\oplus[p_{2}q_{3}]_{\mathbf{R}}.
                \end{align*}
            \end{itemize}

            \medskip
            \textbf{Reconstruction:} Each $\mathbf{X\in\{\mathbf{P}_{1},\mathbf{P}_{2},\mathbf{P}_{3},\mathbf{R}\}}$ opens $[x_{1}x_{2}x_{3}]_{\mathbf{X}}$, such that each $\mathbf{X}$ can locally compute $x_{1}x_{2}x_{3}$.
        \end{minipage}
    }
    \caption{A functionality for privately evaluating a quartic bit conjunction, utilizing only $\mathcal{F}_{\text{TD}}$ and broadcast messaging.}
    \label{fig:F-f3}
\end{figure}

\section{\label{sec:DSC}Discussion}
In this work we defined the problem of QTD, which involves disseminating private binary multiplication triples from entangled quantum resources, even in the presence of a malicious adversary. We introduced a resource graph state and measurement-based protocol that accomplishes this task with provable perfect privacy in an honest majority setting, without any assumptions on classical channels. In contrast to symmetric forms of shared randomness, such as shared key, no one party can know the complete correlation that fixes the triple. Here, multipartite entangled states provide a fundamental advantage. We can encode classical correlations within them that are not uniquely determined by a single local view.

In the wider context of MPC within the pre-processing model, QTD serves a useful offline primitive for privately distributing correlated shared randomness, with assumptions weaker than those in any purely classical setting. Rather than requiring a trusted Dealer that does not know the evaluating function and private peer-to-peer channels, our methodology involves a Referee that can know the function in the question and does not place any requirements on private channels, something that is only possible when utilizing shared randomness obtained from quantum resources. This works demonstrates how an honest Referee and ideal entangled quantum resources, allows an honest minority to compute any function with information-theoretic malicious-secure privacy.

\subsubsection*{Acknowledgments}
This work was supported by the NSF Quantum Leap Challenge Institute on Hybrid Quantum Architectures and Networks (NSF Award No. 2016136). M.G. thanks Shouvik Ghorai, Sarah Hagen, and Selina Nie for helpful discussions. 

\printbibliography

\appendix

\section{\label{sec:APP-Graph states}Obtaining distributed correlations from graph states}
For an arbitrary graph $G=(V,E)$ with vertices $V=\{1,\cdots,n\}$ and edge set $E\subseteq V\times V$, consider the $n$-qubit operator obtained by globally performing a controlled$-Z$ gate, $CZ_{i,j}$, between every $\{i,j\}\in E$, denoted by
\begin{equation}
    U_{G}=\prod_{\{i,j\}\in E}CZ_{i,j}.
\end{equation}
The graph state associated with the graph $G$ is the $n$-qubit state 
\begin{equation}
    \ket{G}=U_{G}\ket{+}^{\otimes V},
\end{equation}
where $\ket{+}=(\ket{0}+\ket{1})/\sqrt{2}$. Note that the stabilizer of $\ket{+}^{\otimes n}$ is generated by a set of $n$ commuting operators $\{X_i\}_{i\in V}$.  Hence, the stabilizer of $\ket{G}$, is then generated by the set of correlation operators $\{K_i\}_{i\in V}$, such that
\begin{align}
    K_{i}&=U_{G}X_{i}U_{G} \nonumber\\
    &=X_{i}\prod_{j\in N_{i}}Z_j =X_{i}\prod_{j\in V}Z_{j}^{\Gamma(i,j)},
\end{align}
which we express either in terms of $N_{i}=\{j|\{i,j\}\in E\}$, the neighborhood of $G$ about qubit $i$, or $\Gamma(i,j)$, the corresponding adjacency matrix of $G$. 

For any $n$-qubit graph state $\ket{G}$, we can generate an orthonormal basis for $\mathbb{C}_2^N$, called the \textit{associated graph basis}. These basis vectors have the form $Z^{\vec{r}}\ket{G}$, with
\begin{equation}
    Z^{\vec{r}}=\prod_{i,\in V}Z_{i}^{r_{i}},
\end{equation}
where $\mathbf{r}=(r_1,r_2,\cdots,r_n)\in\mathbb{Z}_2^n$ forms a conditional phase vector with each bit $r_i$ as the phase information associated with qubit $i$. Let $\vec{r}$ and $\vec{r}'$ be distinct vectors of phase information that differ at position $j$, such that $K_{j}$ anti-commutes with $Z^{\vec{r}}Z^{\vec{r}'}$. Its evident that
\begin{equation}
    \braket{G|Z^{\vec{r}}Z^{\vec{r}'}|G}=\braket{G|Z^{\vec{r}}Z^{\vec{r}'}K_{j}|G}=-\braket{G|K_{j}Z^{\vec{r}}Z^{\vec{r}'}|G}=-\braket{G|Z^{\vec{r}}Z^{\vec{r}'}|G}=0,
\end{equation}
implying that $\{Z^{\vec{r}}\ket{G}\}_{\vec{r}\in\mathbb{Z}_2^n}$ form an orthonormal basis. Hence, the state $Z^{\vec{r}}\ket{G}$ is a unique pure state, defined by the phase information $\{r_{i}\}_{i\in V}$, or equivalently the set of correlation operators $\{(-1)^{r_{i}}K_{i}\}_{i\in V}$.

We will be particularly interested in how the associated basis states transform under the action of local Pauli measurements. For a vector $\vec{r}\in \mathbb{Z}_2^n$ and subset of nodes $\Omega\subseteq V$, let $\vec{r}-\Omega$ denote a vector of length $n-|\Omega|$, obtained from $\vec{r}$ by removing the coordinates in $\Omega$.  Suppose that for an arbitrary associated basis state, a Pauli observable is measured on qubit $i$ and an outcome $m_i\in\{0,1\}$ is recorded. The initial state transforms as follows,
\begin{subequations}
    \begin{align}
        Z_{i}:& \quad Z^{\vec{r}}\ket{G}\mapsto \left(\prod_{j\in N_{i}}  Z_{j}^{m_{i}}\right)Z^{\vec{r}-\{i\}}\ket{G-\{i\}}, \label{Eq:Z-measurement-post-correction}\\
        Y_{i}:& \quad Z^{\vec{r}}\ket{G}\mapsto \left(\prod_{j\in N_{i}} \left(-iZ_{j}\right)^{1/2}Z_{j}^{r_{i}+m_{i}}\right)Z^{\vec{r}-\{i\}}\ket{\tau_{i}(G)-\{i\}}, \label{Eq:Y-measurement-post-correction}\\
        X_{i}:& \quad Z^{\vec{r}}\ket{G}\mapsto (-iY_{j_{0}})^{1/2}Z_{j_{0}}^{r_{j_{0}}+r_{i}+m_{i}} \left(\prod_{j\in N_{i}-N_{j_{0}}-\{j_{0}\}}Z_{j}^{r_{j_{0}}}\right)\left(\prod_{j\in N_{j_{0}}-N_{i}-\{i\}}Z_{j}^{r_{i} + m_{i}}\right) \nonumber\\
        &\quad\quad\quad\quad\quad\quad\quad\quad\quad\quad\quad\quad\quad\quad\quad\quad\quad\quad\quad\quad\quad\quad\quad\quad\quad\quad\quad\quad \times Z^{\vec{r}-\{i\}}\ket{\tau_{j_{0}}(\tau_{i}\circ\tau_{j_{0}}(G)-\{i\})}, \label{Eq:X-measurement-post-correction}
    \end{align}
\end{subequations}
where we discard the register of the measured qubit for brevity. Here, $\tau_{i}(G)$ is the \textit{local complementation} of $G$ at vertex $i$, i.e. $\tau_{i}(G)$ is the graph $(V, E\Delta E(N_i,N_i))$ (where $\Delta$ is the symmetric difference), and $\tau_i(G)-\{i\}$ is the graph obtained by removing $i$ from $\tau_i(G)$ \cite{Hein-2004-Me}. In more detail,
\begin{equation}
    \ket{\tau_{i}(G)}=\left(-iX_i\right)^{1/2}\left(\prod_{j\in N_{i}}\left(iZ_{j}\right)^{1/2}\right)\ket{G}, \label{eq:local-comp}
\end{equation}
describes a set of single-qubit rotations comprising a local complementation. Note that the $Y$ and $X$ post-measurement states can always be transformed back to the associated graph basis by performing the appropriate local unitary operators. The choice of $j_{0}\in N_{i}$ for the $X$ post-measurement state is arbitrary, as each leads to a state equivalent up to a local complementation of $j_{0}'\neq j_{0}$ followed by one on $j_{0}$. In general, we say the measurements in the basis of $Y_{i}$ and $X_{i}$ on qubit $i$ transfers $r_{i}$ to qubits in vicinity of its neighbors within the new graph, while a measurement in the basis of $Z_{i}$ simply deletes $r_{i}$.

Let $S=\langle\{(-1)^{r_{i}}K_{i}\}_{i\in V}\rangle$ be the stabilizer of a graph state $Z^{\vec{r}}\ket{G}$. A \textit{perfectly correlated} set of measurement outcomes occurs, when an element of $S$ is measured, either as a set of local single qubit Pauli measurements or jointly across vertices of the graph. Given a $K_{i}$ that acts non-trivially on only qubits in the subset $\{i,N_{i}\}$ within the graph, measuring each qubit locally in the basis specified by $K_{i}$ leads to a correlation
\begin{equation}
    \Tr[K_{i}Z^{\vec{r}}\ketbra{G}Z^{\vec{r}}]=(-1)^{\bigoplus_{j\in \{i,N_{i}\}}m_{j}}=(-1)^{r_{i}},
\end{equation}
between the set of measurement outcomes $\{m_{j}\}_{j\in\{i,N_{i}\}}$. As $Z^{\vec{r}}\ket{G}$ is a unique basis state $\mathbb{C}_{2}^{n}$, the set of perfect correlations one can obtain from $Z^{\vec{r}}\ket{G}$ is uniquely fixed by $\{r_{i}\}_{i\in V}$, and any possible sum of these elements.



\end{document}